\documentclass[11pt]{article}
\usepackage[english]{babel}
\usepackage{fullpage}
\usepackage[utf8]{inputenc}
\usepackage{natbib}
\usepackage{times}
\usepackage{soul}
\usepackage{url}
\usepackage[hidelinks]{hyperref}
\usepackage[utf8]{inputenc}
\usepackage[small]{caption}
\usepackage{graphicx}
\usepackage{amsmath}
\usepackage{amsthm}
\usepackage{booktabs}
\usepackage{algorithm}
\usepackage{algorithmic}
\usepackage[switch]{lineno}
\newtheorem{example}{Example}
\newtheorem{theorem}{Theorem}

\usepackage{subcaption}
\usepackage{enumerate}
\usepackage[shortlabels]{enumitem}
\usepackage{tikz,tkz-graph,tkz-berge}
\usetikzlibrary{decorations.pathreplacing} 
\usetikzlibrary{calc} 
\usetikzlibrary{fit}
\usepackage{amsthm}
\usepackage{cleveref}
\usepackage{amsmath,amsthm,amsfonts,mathtools}
\newtheorem{proposition}{Proposition}
\newtheorem{observation}{Observation}
\newtheorem{lemma}{Lemma}

\newtheorem{remark}{Remark}

\newtheorem{definition}{Definition}
\newtheorem{fact}{Fact}

\renewcommand{\epsilon}{\varepsilon}

\newcommand{\agents}{\mathcal{N}}

\newcommand{\declarations}{\mathcal{D}}

\newcommand{\SW}{\mathsf{SW}}
\newcommand{\set}[1]{\{#1\}}

\newcommand{\modulus}[1]{\vert #1 \rvert }

\newcommand{\repr}{\textsc{Repr}}
\newcommand{\cut}{Cut}
\newcommand{\noti}{_{-i}}
\newcommand{\parity}{\phi_n}

\newcommand{\OPT}{\textsf{OPT}}
\newcommand{\opt}{\textsf{opt}}
\newcommand{\instance}{\mathcal{I}}

\newcommand{\RR}{\ensuremath{\mathbb{R}}}

\newtheorem{innercustomgeneric}{\customgenericname}
\providecommand{\customgenericname}{}
\newcommand{\newcustomtheorem}[2]{%
	\newenvironment{#1}[1]
	{%
		\renewcommand\customgenericname{#2}%
		\renewcommand\theinnercustomgeneric{##1}%
		\innercustomgeneric
	}
	{\endinnercustomgeneric}
}

\newcustomtheorem{customtheorem}{Theorem}

\newcustomtheorem{customMechanism}{Mechanism}
\newcommand{\mech}{\mathcal{M}}
\newcommand{\declared}{\mathbf{d}}

\newcounter{mechanismCounter}
\newenvironment{mechanism}
{
	\begin{customMechanism}{\text{$\mech_\arabic{mechanismCounter}$}}
		\addtocounter{mechanismCounter}{1}
	}{\end{customMechanism}
}

\DeclareMathOperator{\argmax}{arg\,max}

\title{Non-Obvious Manipulability in Additively Separable and Fractional Hedonic Games}
\author{Diodato Ferraioli$^{1}$, Maria Fomenko$^2$, and Giovanna Varricchio$^3$}
\date{%
	$^1$University of Salerno, Italy\\%
    $^2$Gran Sasso Science Institute,  Italy\\%
	$^3$University of Calabria,  Italy\\%
	dferraioli@unisa.it,
    maria.fomenko@gssi.it,
    giovanna.varricchio@unical.it
	\\[2ex]%
	\today
}

\begin{document}

\maketitle
\begin{abstract}
Hedonic Games are a well-established model for describing the formation of coalitions. In this work, we considered the design of Non-Obviously Manipulable (NOM) mechanisms, that are mechanisms that bounded rational agents may fail to recognize as manipulable, for two relevant classes of succinctly representable Hedonic Games, namely Additively Separable and Fractional Hedonic Games. In these classes, agents have cardinal scores towards other agents, and their preferences towards different coalitions are determined by aggregating these scores. Moreover, the quality of an outcome can also be easily evaluated through these scores by means of the utilitarian social welfare.

We first prove that, when scores can be arbitrary, every welfare-maximizing mechanism is NOM, and, when scores are limited in a continuous interval, then there exist tie-breaking rules making welfare-maximizing mechanisms NOM.

Next, we focus on efficient NOM mechanisms, since there is no known polynomial-time algorithm to compute welfare-maximizing outcomes in the considered classes of hedonic games. 
To this aim, we first prove a characterization of NOM mechanisms that simplifies the class of mechanisms of interest. Then, we design a NOM mechanism returning approximations that essentially match the best-known approximation achievable in polynomial time.
Finally, we turn our attention to discrete scores, and specifically, the case that scores are $\{-x, 0, 1\}$ for $x > 0$. We prove that the ability to design welfare-maximizing NOM mechanisms depends on the magnitude of the scores. In particular, for $x > 1$, we prove that a welfare-maximizing NOM mechanism exists only when $x$ is very large. For $x \leq 1$, instead, we observe that a welfare-maximizing NOM mechanism always exists except when $x$ lies in the interval $[a, b]$ where $a \approx 2/n^2$ and $b \approx 1/n$.  
\end{abstract}

\section{Introduction}
\emph{Hedonic Games} are a well-established model for describing coalition formation (\cite{dreze1980hedonic}). In these games, agents have preferences over possible coalitions, and the goal is to compute a ``good'' partition of the agents, which constitutes the outcome of the game. This topic has been extensively studied in the AI literature, with numerous contributions addressing both structural and computational aspects (e.g., \cite{elkind2009hedonic,aziz2015welfare,igarashi2016hedonic,aziz2019fractional,flammini2021distance,flammini2022strategyproof,flammini2025non}).

Within this line of research, two main approaches can be distinguished: a decentralized approach, where agents autonomously select their coalition, and a centralized one, where agents report their preferences to a designer who determines the final outcome.
While in the first approach, it is assumed that the attained partition satisfies some form of stability against either individual or group deviations (\cite{Bloch2011,Feldman15,gairing2010computing,Bogomolnaia02,banerjee2001core,elkind2009hedonic,igarashi2016hedonic}), in the second, it is instead necessary to design algorithms, a.k.a., \emph{mechanisms}, incentivizing agents to report their real preferences (\cite{flammini2021strategyproof,flammini2022strategyproof,klaus2023core,dimitrov2004enemies,dimitrov2006simple,rodriguez2009strategy,varricchio2023approximate}).

A standard requirement in mechanism design is \emph{strategyproofness} (SP), which guarantees that truthful reporting is a dominant strategy for every agent. However, strategyproofness is often too restrictive as it implicitly assumes fully rational agents capable of evaluating all possible preference profiles of others, which is computationally infeasible in many settings and cognitively unrealistic in others. In fact, empirical evidence shows that agents frequently lack contingent reasoning abilities and cannot perform exhaustive case analysis (\cite{charness2009origin,esponda2014hypothetical,ngangoue2021learning}).

Moreover, asking for strategyproof mechanisms against fully rational agents usually results in a significantly poor quality of the computed outcomes. In particular, \cite{ozyurt2009general} show that it is, in general, impossible to have non-dictatorial strategyproof mechanisms. Similar impossibility results have also been provided in more restricted settings and specific coalitions' quality measures (\cite{amanatidis2017truthful,brandl2018proving}). For example, \cite{flammini2021strategyproof} proved that no strategyproof mechanism can return a bounded approximation of the optimal partition, for hedonic games whose preferences are succinctly representable and the outcomes are evaluated with respect to the utilitarian social welfare.

These limitations naturally raise the question of whether better partitions can be produced by requiring the mechanism to be truthful against bounded rational agents. Two main kinds of bounded rationality have attracted most of the attention recently. In {\em Risk-Averse Truthful} (RAT) mechanisms design (\cite{bu2023existence}), the agents do not know the private information (type) of other agents, and they are assumed to adopt a risk-averse behavior. Indeed, classical strategyproofness assumes that an agent is risk-prone in the sense that she deviates from truthfulness as long as this deviation is profitable for at least one profile of other agents' declarations. 

In contrast, in RAT mechanisms, an agent deviates only if a misreport is profitable for some profile of others’ declarations and never results in a lower utility than truthful reporting.
Observe that RAT still requires evaluating the profitability of a deviation among all possible declarations of other agents, 
which is not affordable.
Moreover, as we show below, RAT mechanisms are insufficient to guarantee 
bounded approximation compared to the optimal coalition structure.

In \emph{Non-Obviously Manipulable} (NOM) mechanisms design (\cite{troyan2020obvious}),
agents do not consider all possible declarations of the others, but only the best and worst cases. 
Specifically, in NOM mechanisms, an agent has no incentive in
misreporting their preferences (i.e., to {\em manipulate})
only if: (i) the best outcome (over all possible reports of the other agents) achievable under the misreport is no better than the best outcome achievable under truthful reporting; and (ii) the worst outcome achievable under the misreport is no better than the worst outcome achievable under truthful reporting.
This choice for aggregating the possible realizations is motivated by the fact that people often fail to recognize profitable manipulations, for example, in school choice and two-sided matching problems~(\cite{dur2018identifying,pathak2008leveling}), whereas reasoning in terms of best- and worst-case outcomes is comparatively simpler.

Non-obviously manipulable mechanisms have been recently designed for several problems: school choice, two-sided matching, auctions, bilateral trade (\cite{troyan2020obvious}), voting (\cite{aziz2021obvious,arribillaga2024obvious}), fair division (\cite{ortega2022obvious,psomas2022fair}), and single-parameter domains (\cite{archbold2023non_aamas,archbold2023non_ijcai,archbold2024willy}). For many of these, it has been shown that NOM mechanisms can produce much better outcomes than SP mechanisms. For example, in \cite{psomas2022fair}, it has been shown that, in the problem of fair allocation of indivisible goods, replacing strategyproofness with non-obvious manipulability allows one to design a Pareto-efficient and non-dictatorial mechanism as well as a mechanism that guarantees envy-freeness up to one item and thus circumvent existing negative results for SP mechanisms. Very recently, similar results have been obtained by \cite{flammini2025non} for hedonic games with friends' appreciation preferences, a special class of additively separable hedonic games. 
This suggests that similar positive results may extend to more expressive classes of hedonic games.

\paragraph{Our Contribution.}
In this work, we focus on two well-studied subclasses of hedonic games with succinct preference representation: \emph{Fractional Hedonic Games} (FHGs) (\cite{aziz2019fractional}) and \emph{Additively Separable Hedonic Games} (ASHGs) (\cite{hajdukova2004coalition}). Here, preferences over coalitions are expressed by utilities computed from the scores assigned by an agent to every other coalition member. This allows evaluating the partition's quality by using \emph{utilitarian social welfare}, the sum of agents' utilities for the outcome.

We first show that, while no SP mechanism may return even a bounded approximation
of the optimal social welfare (\cite{flammini2021strategyproof}), every mechanism that returns such an optimum is NOM. Unfortunately, this result does not completely settle the problem for two reasons: First, in both FHGs and ASHGs, computing the social welfare maximizing partition is a computationally hard problem (\cite{flammini2022strategyproof,aziz2015welfare}), and hence one would be interested in evaluating the approximation of the social welfare achievable by mechanisms that are both NOM and computationally efficient. Moreover, optimal mechanisms may fail to be NOM as soon as scores are restricted to limited intervals.

We address the first issue by proving that an efficient NOM mechanism that returns an $n$-approximation of the optimal social welfare in ASHGs, where $n$ is the number of agents, and a $2$-approximation in FHGs exists. We stress that for ASHGs the bound is tight, as an approximation better than $n$ cannot be computed in poly-time, unless P=NP, even if one does not care about manipulability (\cite{flammini2022strategyproof}).

On the route to achieve this, we prove another result that may be of independent interest: We show that when one cares about NOM mechanisms with bounded approximation, it is without loss of generality to look at ``natural'' mechanisms, i.e., mechanisms that determine whether two agents must belong to the same coalition, only by considering how much they like/dislike each other (in total), regardless of the specific mutual scores.
This characterization also allows us to show that, while  
not every optimal mechanism is NOM, there always exists an optimal mechanism that is NOM, even if scores are bounded in a continuous interval. 

This result no longer holds true as soon as we turn our attention to discrete values: for example, \cite{flammini2025non} proved that in ASHGs, if scores are in $\set{-n,1}$, then no optimal mechanism is NOM. In this work, we elaborate more on these settings by providing a(n almost\footnote{A  tiny gap for $x \in \left[\frac{1}{n+\sqrt{2n+1}}, \frac{1}{n+1}\right)$ is left for which we do not know whether an optimal and NOM mechanism exists.}) full characterization of the existence of optimal NOM mechanisms for ASHGs when scores can take values in $\{-x, 0, 1\}$, with $x> 0$. 
Specifically, we show that for $x > 1$, we prove that optimal NOM mechanisms exist only for very large value of $x$, namely $x \geq 2n-3$. As for $x \leq 1$, we have that there are two phase transitions at $\tau_1$ and $\tau_2$, with $\tau_1 \approx \frac{2}{n^2-3}$ and $\tau_2 \approx 1/n$: an optimal NOM mechanism is indeed proved to exists only whenever either $x \leq \tau_1$ or $x \geq \tau_2$.

\section{Notation and Preliminaries}
In this section, we introduce the notations and outline some preliminary results.

\subsection{Hedonic Games}
In \emph{Hedonic Games} (HGs), we are given a set of $n$ agents, denoted by $\agents=\set{1,\dots,n}$, and the goal is to split them into disjoint coalitions, that is, to provide a \emph{partition} $\pi=\set{C_1,\dots, C_k}$ such that $\cup_{h=1}^k C_h = \agents$. Such a partition $\pi$ is also called \emph{outcome} or \emph{coalition structure}. For any agent $i\in\agents$, we denote by $\pi(i)$ the coalition $i$ is assigned to in $\pi$. The \emph{grand coalition} refers to a coalition containing all the agents; a \emph{singleton coalition} is a coalition consisting of only one agent.
A defining feature of HGs is that agents have preferences over outcomes solely based on the coalition they belong to, and independent of how the other agents are grouped. We denote the set of all possible coalitions of $i$ by $\agents_i= \set{C\subseteq \agents\ \vert\ i\in C}$.

Particularly relevant are HG subclasses where agents express preferences through utilities having succinct representation, meaning that the memory required to store the preference information is polynomial in the size of the instance. 

In \emph{fractional hedonic games} (FHGs) and \emph{additively separable hedonic games} (ASHGs), for example, agents assign individual scores to other agents, and their utility for a coalition depends on the sum or average of these scores over coalition members.
More formally, each agent $i$ has a weight function $w_i:\agents \rightarrow \mathbb{R}$ that determines the value that agent $i$ assigns to any other participant. We always assume that $w_i(i)=0$ and denote $w_{ij}=w_i(j)$.

In ASHGs and FHGs, the \emph{utility} that agent $i$ derives from being in a coalition $C\in \agents_i$ is given by 
\begin{align*}
   u_i(C)= \sum_{j\in C} w_{ij} \;\;  \text{ and } \;\;  u_i(C)= \frac{\sum_{j\in C} w_{ij}}{\modulus{C}}\, ,  \text{ respectively.}
\end{align*}

Note that $u_i(C)$ is well-defined only if $C \in \agents_i$, as agents derive utility exclusively from the coalition they are part of.

In principle, $w_{ij}$ can be any value in $\mathbb{R}$.  However, we may also consider specific restrictions on these weights. We talk about:

\begin{itemize}
    \item \emph{arbitrary} weights, if $w_{ij}\in\mathbb{R}$, for all $i,j\in\agents$;
    \item \emph{non-negative} weights, if $w_{ij}\in\mathbb{R}_{\geq 0}$, for all $i,j\in\agents$;
    \item \emph{bounded} weights, if $w_{ij}\in [-1,1]$, for all $i,j\in\agents$;
    \item \emph{general duplex} weights, if $w_{ij}\in \set{-x,0,1}$, for all $i,j\in\agents$ and $x\in (0, \infty)$ (in case $x=1$ we simply say {\em duplex}).
\end{itemize}
The distinction in arbitrary, non-negative, bounded, and duplex weights has been presented by \cite{flammini2021strategyproof}. We introduce the notion of general duplex valuations to include the well-known class of friends and enemies games (\cite{dimitrov2004enemies}). 
If we do not specify otherwise, we assume the values are arbitrary real numbers.

\paragraph{Graph Representation and Flattened Graph.}
Both ASHGs and FHGs have a very convenient representation through a directed and weighted graph
$\Vec{G} = (\agents, E, w)$, where $\agents$ are the vertices and a directed edge $(i, j) \in E$ for $i \neq j$ has the weight $w_{ij}$, denoting the value $i$ has for agent $j$. See, for example, \Cref{subfig:firstInstance,subfig:secondInstance}. We may refer to $\Vec{G}$ as the graph corresponding to a given instance, and we always consider it to be complete: an edge exists even if its weight is $0$. As the set of edges is always the same, instead of $\Vec{G} = (\agents, E, w)$ we will hereafter simply write $\Vec{G} = (\agents, w)$.

Given two disjoint coalitions $C_1, C_2 \subseteq \agents$, we may consider the \emph{cut} induced in the graph $\Vec{G}$ by these two coalitions, that is, 
\begin{align*}
   \cut(C_1,C_2)=\set{(i,j), (j,i) \, \vert \, i\in C_1, \, j\in C_2} \, . 
\end{align*}

We also define the value of the cut $\cut(C_1,C_2)$ as 
\begin{align*}
    w(\cut(C_1,C_2)) = \sum_{ i\in C_1, \, j\in C_2} \left( w_{ij} + w_{ji}\right)\, .
\end{align*}

Moreover, we denote by 
%$\delta^+(i)=\set{j\neq i \,\vert\, w_{ij}> 0}$ and $\delta^-(i)=\set{j\neq i \,\vert\, w_{ij}<0}$ 
\begin{align*}
   \delta^+(i)=\set{j\neq i \,\vert\, w_{ij}> 0} && \text{ and } && \delta^-(i)=\set{j\neq i \,\vert\, w_{ij}<0}
\end{align*}
the set of positive and negative neighbors of $i$, respectively.

We will also consider an undirected version of $\Vec{G}$ where the weight of an edge is given by the sum of the mutual values.  
More formally, we say that a weighted and undirected graph $G=(\agents, \hat{w})$ is the \emph{flattened graph} of a directed and weighted graph $\Vec{G}$, if it has the same vertices as $\Vec{G}$ and, for each pair of directed edges $(i,j)$ and $(j,i)$ and their weights $w_{ij}$ and $w_{ji}$ in $\Vec{G}$, respectively, there is an undirected edge $e=\set{i,j}$ of weight $\hat{w}(e)=\hat{w}(i,j)=w_{ij}+ w_{ji}$ in $G$. Notice that an ASHG or FHG instance is uniquely determined by the directed graph $\Vec{G}$ we described above; however, different instances may have the same flattened graph $G$. For example, the graph depicted in \Cref{subfig:flattened} is the flattened version of the graph representation provided in   \Cref{subfig:firstInstance} and \Cref{subfig:secondInstance}.

\begin{figure}[tb]\centering 
\def \variable {2.2cm}
	\begin{subfigure}{0.32\textwidth}
		\centering
		\begin{tikzpicture}[->,>=stealth',auto,node distance=\variable,
		thick,main node/.style={circle,draw,font=\sffamily\bfseries}]
		\node[main node] (2) {$2$};
		\node[main node] (1) [above left of=2, yshift=0.5cm] {$1$};
		\node[main node] (3) [above right of=2,  yshift=0.5cm] {$3$};
		\path[every node/.style={font=\sffamily\small}]
		(2) edge [bend left =10]  node {$2$} (1)
            (1) edge [bend left =10]  node [yshift=-0.2cm] {$3$} (2)
		(3) edge [bend left =10] node {$0$} (2)
		(2) edge [bend left =10] node [yshift=-0.2cm] {$2$} (3)
		(1) edge [bend left =10]  node {$-7$} (3)
		(3) edge [bend left =10]  node {$4$} (1)	
		;
		\end{tikzpicture}
		\caption{Example of a graph representation $\Vec{G}$.}
	\label{subfig:firstInstance}
	\end{subfigure}
	\quad
	\begin{subfigure}{0.32\textwidth}
		\centering
		\begin{tikzpicture}[->,>=stealth',auto,node distance=\variable,
		thick,main node/.style={circle,draw,font=\sffamily\bfseries}]
		\node[main node] (2) {$2$};
		\node[main node] (1) [above left of=2, yshift=0.5cm] {$1$};
		\node[main node] (3) [above right of=2,  yshift=0.5cm] {$3$};
		\path[every node/.style={font=\sffamily\small}]
		(2) edge [bend left =10]  node {$1$} (1)
            (1) edge [bend left =10]  node [yshift=-0.2cm] {$4$} (2)
		(3) edge [bend left =10] node {$1$} (2)
		(2) edge [bend left =10] node [yshift=-0.2cm] {$1$} (3)
		(1) edge [bend left =10]  node {$1$} (3)
		(3) edge [bend left =10]  node {$-4$} (1)	;
		\end{tikzpicture}
		\caption{Example of a graph representation $\Vec{G}$.}
		\label{subfig:secondInstance}
	\end{subfigure}
	\quad
	\begin{subfigure}{0.3\textwidth}
		\centering
		\begin{tikzpicture}[-,>=stealth',auto,node distance=\variable,
		thick,main node/.style={circle,draw,font=\sffamily\bfseries}]
		\node[main node] (2) {$2$};
		\node[main node] (1) [above left of=2, yshift=0.5cm] {$1$};
		\node[main node] (3) [above right of=2,  yshift=0.5cm] {$3$};
		\path[every node/.style={font=\sffamily\small}]
		(3) edge node {$2$} (2)
            (2) edge node {$5$} (1)
            (1) edge node {$-3$} (3);
		\end{tikzpicture}
		\caption{The corresponding flattened graph.}
		\label{subfig:flattened}
	\end{subfigure}
	\caption{Example of two game instances having the same flattened graph.}
	\label{fig:graphReprExample}
\end{figure}

\paragraph{Proportional Graphs.}
Given two undirected and weighted graphs $G$ and $G'$ with the same set of vertices, and having edge weights $\hat{w}$ and $\hat{w}'$, we say that $G$ and $G'$ are \emph{proportional} if there exists $\lambda > 0$ such that $\hat{w}(i,j)= \lambda \cdot \hat{w}'(i,j)$ holds true for each pair of vertices $i\neq j$.  
Notice that proportionality is an equivalence relation over the set of weighted and undirected graphs; we write $G\sim_PG'$ to denote that the two graphs are proportional. 
We often slightly abuse notation, and we will simply say that two instances, namely $\instance, \instance'$, are proportional, and write $\instance\sim_P\instance'$, if the flattened graphs modeling these instances are proportional.
In \Cref{sec:characterization}, proportionality will turn out to be a useful property to characterize mechanisms resistant to obvious manipulations.

\subsection{Strategyproofness and Manipulability}
To compute desirable partitions of the agents, we must first elicit their preferences. However, the values $w_{ij}$ are private information of agent $i$, and therefore must be reported by the agent.
Let $\declared = (d_1, \dots, d_n)$ be the vector of agents' declarations, where the declaration $d_i$ contains the information disclosed by agent $i$. We assume direct revelation, which means, $d_i = \{d_{ij}\}_{j \in \agents \setminus \{i\}}$, where $d_{ij}$ represents the value that agent $i$ declares for agent $j$. We denote by $\declarations$ the space of all possible declarations.  For notational convenience, $\declared_{-i}$ denotes the declarations of all agents except $i$ and $\declarations_{-i}$ the set of all feasible $\declared_{-i}$. Similarly, we denote by $\declarations_i$ the space of feasible $d_i$.

A {\em mechanism}, which is an algorithm computing an outcome upon agents' declarations, is denoted by $\mech$, while $\mech(\declared)$ is its output for a given declaration $\declared\in\declarations$. 
We write $\mech_i(\declared)$ to denote the coalition $i$ is assigned to.
Since the agents are strategic, an agent $i$ could possibly declare $d_i\neq w_i$, where $w_i$ contains the real values assigned by agent $i$ and will henceforth also be called \emph{real type}.  
In this context, a key challenge is to design
mechanisms
encouraging agents to behave truthfully, i.e., to submit a declaration $d_i = w_i$.
The most desirable and widely studied characteristic for such a kind of mechanism is \emph{strategyproofness}.

\begin{definition}[Strategyproofness]\label{def:SP}
    A mechanism $\mech$ is said to be \emph{strategyproof} (SP) if for every $i\in\agents$, real type $w_i$, and for any declaration of the other agents $\declared_{-i}$
    \begin{equation}\label{eq:strategyproofness}
         u_i(\mech(w_i, \declared_{-i})) \geq u_i(\mech(d_i, \declared_{-i}))
    \end{equation}
    holds true for any possible declaration $d_i\neq w_i$ of agent $i$.

    In turn, a mechanism is said to be \emph{manipulable} if there exists an agent $i$, a real type $w_i$, and a declaration $d_i\neq w_i$ s.t.\ Eq.~\ref{eq:strategyproofness} does not hold true. Such a $d_i$ is called a \emph{manipulation}.
\end{definition}
Notice that strategyproofness requires that for {\em every} possible declaration of the other agents, there is no convenience in misreporting. 

Due to this demanding requirement, strategyproof mechanisms have been shown to be highly inefficient; for this reason, mechanisms satisfying milder conditions have been introduced with the scope of obtaining more efficient outcomes. 
Therefore, in this work, we will discuss
the less stringent conditions of \emph{risk-averse truthfulness} and {\em non-obvious manipulability}.

\begin{definition}[Risk-Averse Truthfulness]\label{def:RAT}
    A mechanism $\mech$ is said to be \emph{risk-averse truthful} (RAT) if for every $i\in\agents$, real type $w_i$, and any other declaration $d_i$ one of the following two conditions hold true:
    \begin{enumerate}
        \item\label{RAT:sup} there is $\declared_{-i}$ such that $u_i(\mech(w_i, \declared_{-i})) > u_i(\mech(d_i, \declared_{-i}))$, or
        \item\label{RAT:inf} for all $\declared_{-i}$ it holds that $u_i(\mech(w_i, \declared_{-i})) \geq u_i(\mech(d_i, \declared_{-i}))$.
    \end{enumerate}
In case there exist $i$, $w_i$, and $d_i$ such that both the aforementioned conditions are violated, then $\mech$ is \emph{risk-averse manipulable} and $d_i$ is a \emph{risk-averse manipulation}.
\end{definition}

In other words, the mechanism is either strategyproof (Condition 2) or, for any possible misreport, there exists at least one declaration of the others that makes it not profitable.
Note that this also means that SP $\Rightarrow$ RAT.
In fact, under this notion, significantly fewer deviations from truthful reporting are allowed. 
Nevertheless, we show below that RAT mechanisms are still unable to address the inefficiency issues of strategyproofness.

For this reason, the main focus of this work will be on the following condition.
\begin{definition}[Non-Obvious Manipulability]\label{def:NOM}
    A mechanism $\mech$ is said to be \emph{non-obviously manipulable} (NOM) if for every $i\in\agents$, real type $w_i$, and any other declaration $d_i$ the following two conditions hold true:
    \begin{enumerate}
        \item\label{NOM:sup} $\sup\limits_{\declared_{-i}} u_i(\mech(w_i, \declared_{-i})) \geq \sup\limits_{\declared_{-i}} u_i(\mech(d_i, \declared_{-i}))$, and 
        \item\label{NOM:inf} $\inf\limits_{\declared_{-i}} u_i(\mech(w_i, \declared_{-i})) \geq \inf\limits_{\declared_{-i}} u_i(\mech(d_i, \declared_{-i}))$.
    \end{enumerate}
In case there exist $i$, $w_i$, and $d_i$ such that one of the aforementioned conditions is violated, then $\mech$ is \emph{obviously manipulable} and $d_i$ is an \emph{obvious manipulation}.
\end{definition}

In other words, to determine if a mechanism $\mech$ is NOM for an agent $i$, we need to consider the best- and the worst-case outcomes for $i$ (according to her truthful preferences) provided by $\mech$ when $i$ declares her true type $w_i$ and any other possible declaration $d_i\neq w_i$. In particular, neither the best nor the worst possible outcome when declaring $w_i$ can be strictly worse than the best or the worst outcome, according to her true preferences, attainable when declaring $d_i$.
See \Cref{subfig:NOM} for an informal graphical representation.

\begin{figure}[tb]
    \centering
    \begin{subfigure}[h]{0.48\textwidth}
     \centering
    \begin{tikzpicture}[scale=0.4, baseline={(-5,0)}]
    \draw [-stealth] (-5,0) -- (9,0) node [right]  {$\RR$}; %Axis
    \draw[loosely dotted, thick] (-3,2) -- (5,2) ;  
    \draw[loosely dotted, thick] (0,4) -- (7,4);
    \draw[fill=blue] (-3,2) circle (0.2);
    \draw[fill=red] (5,2) circle (0.2);
    \draw[fill=blue] (0,4) circle (0.2);
    \draw[fill=red] (7,4) circle (0.2);
    \node () at (11,2) {every $d_i\neq t_i$};
    \node () at (11,4) {$d_i=t_i$};
    \node () at (2,-0.7) {$u_i(\mech(d_i, \declared_{-i}))$};
    \end{tikzpicture}
    \caption{Non-obvious Manipulability is satisfied by a mechanism $\mech$ for a given agent $i$.}
    \label{subfig:NOM}
    \end{subfigure}
    \hfill
    \begin{subfigure}[h]{0.48\textwidth}
     \centering
    \begin{tikzpicture}[scale=0.4, baseline={(-5,0)}]
    \draw [-stealth] (-5,0) -- (9,0) node [right]  {$\RR$}; %Axis

    \draw[loosely dotted, thick] (-3,2) -- (5,2) ;
    
    \draw[loosely dotted, thick] (0,4) -- (7,4);
    \draw[fill=blue] (-3,2) circle (0.2);
    \draw[fill=red] (5,2) circle (0.2);
    \draw[fill=blue] (0,4) circle (0.2);
    \draw[fill=red] (7,4) circle (0.2);
    \node (A) at (-1.5,.7)  {$\overline{\declared}_{-i}$};
    \node (B) at (4,2)  {};
    \node (C) at (2,4)  {};
    \draw [->] (-1.4,.8) to (B);
    \draw [->] (-1.4,.8) to (C);
    \draw[fill=gray]  (4,2) circle (0.2);
    \draw[fill=gray]  (2,4) circle (0.2);
    \node () at (11,2) {some $d_i\neq t_i$};
    \node () at (11,4) {$d_i=t_i$};
    \node () at (2,-0.7) {$u_i(\mech(d_i, \declared_{-i}))$};
    \end{tikzpicture}
    \caption{Strategyproofness is violated at the declaration $\overline{\declared}_{-i}$ of the others when agent $i$ reports $d_i$.}
    \label{subfig:NOMvsSP}
    \end{subfigure}
    \caption{We display, in increasing order, the utilities attainable by an agent $i$ from the outcomes of mechanism $\mech$ when she declares $d_i$, as $\declared_{-i}$ varies. In \Cref{subfig:NOM}, we show what happens if a mechanism $\mech$  is NOM. In \Cref{subfig:NOMvsSP}, we show that even though on a declaration $d_i$ both Conditions~\ref{NOM:sup} and~\ref{NOM:inf} are satisfied, the mechanism $\mech$ can still be manipulable for some $\overline{\declared}_{-i}$, and $d_i$ is a manipulation.}
    \label{fig:NOMvsSP}
\end{figure}

Clearly, for any game class, SP $\Rightarrow$ NOM. That is, SP is more demanding than NOM as it requires verifying for every declaration of the others if there exists a profitable manipulation. 
In fact, strategyproofness compares outcomes pointwise: for a given $d_i$, the existence of a profile $\overline{\declared}_{-i}$ that yields a strictly higher utility already constitutes a manipulation. See \Cref{subfig:NOMvsSP} for an example where NOM is possibly satisfied, but SP is violated.
However, the agents might not know upfront what the others' declarations are. This is circumvented by NOM.

In order to work with non-obvious manipulability, it is necessary to understand the structure of outcomes guaranteeing such a property. We here provide a preliminary result in this direction: a simple but still useful sufficient condition for a mechanism being NOM. A finer characterization will be given later.

To this aim, let us introduce some necessary notation. 
Given a mechanism $\mech$, for any $d_i\in\declarations_i$, let 
\[
Coal_i^\mech(d_i)= \set{\mech_i(d_i, \declared_{-i}) \, \vert \, \declared_{-i} \in \declarations_{-i}} 
\]
be the set of possible coalitions for $i$ if she declares $d_i$.

\begin{proposition}\label{obs:sameSetIsNOM}
    Given a mechanism $\mech$, if for each $i\in \agents$ and $d_{i}\in\declarations_{i}$, given the truthful declaration $w_i$,
    if
    $Coal_i^\mech(w_i)= Coal_i^\mech(d_i)$,
    then, $\mech$ is NOM.
\end{proposition}
\begin{proof}
  If $Coal_i^\mech(w_i)= Coal_i^\mech(d_i)$, it means that, whatever the declaration of agent $i$ is, the set of coalitions $i$ may end up with remains the same. Then, the best/worst outcome remains the same as well, and the conditions (1) and (2) in \Cref{def:NOM} are always satisfied with the equality, and thus NOM holds true.
\end{proof}
 Proposition~\ref{obs:sameSetIsNOM} provides a simple yet useful sufficient condition for determining whether a mechanism is NOM.

 \subsection{Optimal and Approximate Solutions}
Strategyproofness (SP), and consequently the weaker requirements of Risk-Aware Truthfulness and Non-Obvious Manipulability, is not inherently difficult to satisfy: a mechanism that always outputs the same fixed outcome, regardless of the agents’ declarations, is trivially strategyproof (and hence, RAT and NOM), as agents have no influence over the result and therefore no opportunity to manipulate it. However, while immune to manipulation, this type of mechanism can be highly inefficient: the fixed outcome may be suitable for some instances but inappropriate for others. For this reason, we also require mechanisms to be \emph{efficient}, adapting the outcome to the input preferences in some meaningful way. In our work, efficiency is fulfilled by outcomes that guarantee high social welfare.

With the agents' preferences being expressed by means of utilities, we can evaluate the quality of outcomes through a \emph{social welfare} (SW) function. Specifically, let $\pi$ be any outcome;  the (utilitarian) SW of $\pi$ is given by $\SW(\pi)= \sum_{i\in\agents} u_i(\pi(i))$.
We call \emph{social optimum}, or simply the optimum, any outcome $\OPT$ in $\argmax_{\pi}\SW(\pi)$ and we denote by $\opt$ the value $\SW(\OPT)$.

In the considered classes of games, computing an optimal outcome is often intractable. Therefore, we may focus on finding approximate solutions. 

Let $\opt_\declared$ denote the value of the social optimum for the declaration $\declared$, i.e., for ASHG, \[\opt_\declared = \max_{\pi} \sum_{i \in \agents} \sum_{j \in \pi(i)} d_{ij},\] 
and, for FHG,
\[\opt_\declared = \max_{\pi} \sum_{i \in \agents} \frac{\sum_{j \in \pi(i)} d_{ij}}{|\pi(i)|}.\] 

\begin{definition}[Bounded approximation property (BAPX)]
 A mechanism $\mech$ satisfies the {\em bounded approximation property} (BAPX) if there exists a real-valued and bounded function $r^{\mech}(\cdot)$ such that for every number $n$ of agents it holds that
 \[
 1 \leq \sup\limits_{\declared\in\declarations} \frac{\opt_\declared}{\SW(\mech(\declared))} \leq r^{\mech}(n)
 \]
 We denote by $r^\mech$ the approximation ratio of a mechanism $\mech$.\footnote{We here assume that ``$\frac{0}{0}=1$''. Hence, if $\opt_\declared=0$, then a bounded approximation mechanism $\mech$ must guarantee that $\SW(\mech(\declared) = 0$, and thus optimal.}
\end{definition}

Here, we focus on mechanisms that return solutions that guarantee a non-negative $\SW$. Note that this is without loss of generality since there is at least one non-negative solution (namely, each agent in a singleton), and hence also the optimum is non-negative.

A mechanism that for any input instance returns an optimum partition is said to be {\em optimal}.

\section{Negative Results for RAT mechanisms}
We already mentioned that strategyproofness turned out to be extremely inefficient in terms of approximation to the social optimum (\cite{flammini2021strategyproof}). We here extend these results by showing that they hold even for risk-aware truthfulness.
\begin{proposition}\label{obs:OptNotSP}
    For both ASHGs and FHGs with arbitrary or bounded weights, no RAT mechanism is BAPX.
\end{proposition}
\begin{proof}
The following arguments apply to ASHGs but easily extend to FHGs by dividing the utilities of the agents by the coalition size. 

\begin{figure}[t!]%
	\def \variable {3cm}
    \centering
	\begin{subfigure}[t]{0.45\textwidth}
		\centering
		\begin{tikzpicture}[->,>=stealth',auto,node distance=\variable,
		thick,main node/.style={circle,draw}, box/.style = {draw,dashed,inner sep=12pt,rounded corners=5pt}]
		\node[main node] (1) {$1$};
		\node[main node] (2) [right of=1] {$2$};
            \node[box, fit=(1)(2)] {};
		\path[every node/.style={font=\sffamily\small}]
		(1) edge[bend left=10] node {$d_{12}$} (2)
		(2) edge [bend left=10] node {$-\epsilon$} (1);
		\end{tikzpicture}
		\subcaption{Truthful declaration.}
		\label{fig:noBAPXisSP1}
	\end{subfigure}
	\begin{subfigure}[t]{0.45\textwidth}
		\centering
		\begin{tikzpicture}[->,>=stealth',auto,node distance=\variable,
		thick,main node/.style={circle,draw}, box/.style = {draw,dashed,inner sep=12pt,rounded corners=5pt}]
		\node[main node] (1) {$1$};
		\node[main node] (2) [right of=1] {$2$};
            \node[box, fit=(1)] {};
            \node[box, fit=(2)] {};
		\path[every node/.style={font=\sffamily\small}]
		(1) edge[bend left=10] node {$d_{12}$} (2)
		(2) edge [bend left=10] node {$-1$} (1);
		\end{tikzpicture}
		\subcaption{Manipulation of Agent $2$.}
		\label{fig:noBAPXisSP2}
	\end{subfigure}
	\caption{No BAPX is RAT. Dashed lines represent the only BAPX partitions for instances when $\epsilon < b_1(2) < 1$: in these instances, agent 2 has an advantage in deviating, while in all remaining instances the BAPX partitions are the same in the two instances, making the deviation also risk-averse.}
	\label{fig:noBAPXisSP}
\end{figure} 
Consider a game instance $\instance$ with two agents. Let $d_{12}$ denote the declaration of agent $1$ regarding agent $2$, and assume that $w_{21} = -\epsilon$ is the true weight of agent $1$ towards agent $2$. The graph corresponding to this instance is depicted in \Cref{fig:noBAPXisSP1}.

Suppose first that agent $2$ reports truthfully.
If agent $1$ declares $d_{12} < \epsilon$, the only optimal solution consists of the two agents being split into singletons, with social welfare equal to $0$. Hence, any bounded approximation algorithm should return the optimum, which guarantees utility $0$ to agent $2$.
If agent $1$ declares $d_{12} = \epsilon$, both solutions, namely the singleton coalitions and the grand coalition, have social welfare $0$. Hence, a bounded approximation mechanism can indifferently return any of these solutions. In any case, the singleton coalition guarantees utility $0$ to agent $2$, while the grand coalition guarantees utility $-\epsilon$.
If agent $1$ declares $d_{12} > \epsilon$, the optimal solution is to return the grand coalition, with social welfare strictly larger than $0$. Since the only alternative outcome is to return singleton coalitions with social welfare $0$, every bounded approximation mechanism must return the grand coalition, and hence must guarantee utility $-\epsilon$ to agent $2$.

Assume now that agent $2$ misreports her value towards $1$ by declaring $d_{21} = -1$, as shown in \Cref{fig:noBAPXisSP2}. Similarly to the previous case, we have that: 
if agent $1$ declares $d_{12} < 1$, then any bounded approximation algorithm would guarantee utility $0$ to agent 2; 
if agent $1$ declares $d_{12} = 1$, then a bounded approximation mechanism can guarantee either utility 0 or utility $-\varepsilon$ to agent $2$;
If agent $1$ declares $d_{12} > 1$, then every bounded approximation mechanism must guarantee utility $-\epsilon$ to agent $2$.

In conclusion, if agent $2$ declares $d_{21} = -1$ then, for every declaration $d_{12} \in (\epsilon, 1)$ of agent $1$, agent $2$ will get a higher utility than when truthfully reporting, while she receives the very same utility for every remaining declaration of agent $2$. Hence, if a mechanism $\mech$ is BAPX, in instance $\instance$, declaration $d_{21} = -1$ is a RAT manipulation. Therefore, we can conclude $\mech$ is not RAT, and the thesis follows.
\end{proof}

We also recall that for ASHGs with non-negative weights, the mechanism that always returns the grand coalition is optimal and SP (and hence RAT and NOM) (\cite{flammini2021strategyproof}). Instead, for FHGs with non-negative weights, we know that no $n/2$-approximate  and SP mechanism exists.

\section{Optimality and NOM}
NOM stands in contrast to the impossibility results described in the previous section. 
Before diving into the technical details, let us give a simple example of how the negative result provided by \Cref{obs:OptNotSP} can be circumvented by NOM. Notice that the negative result holds even when $n=2$; with the next example, we show that a NOM and optimal mechanism always exist when $n=2$, both for ASHGs and FHGs with arbitrary weights.

\begin{example}\label{ex:OPTisNOMn=2}
    Consider an ASHGs or FHGs instance with arbitrary weights where there are only two agents.
    Let $d_{12}$ denote the value declared by agent $1$ towards agent $2$, and $d_{21}$ denote the value declared by agent $2$ towards agent $1$. An optimal mechanism will always place the agents in the same coalition whenever $d_{12}+d_{21}>0$,  and in separate coalitions whenever $d_{12}+d_{21}<0$. Now consider agent $1$.  Regardless of agent $1$’s declaration, when considering an optimal mechanism, both outcomes $\set{\set{1,2}}$ and $\set{\set{1},\set{2}}$ are possible, depending on what value is declared by agent $2$. In other words, for a fixed $d_{12}$ there always exists a declaration $d_{21}$ that makes the summation $d_{12}+d_{21}$ either positive or negative, and hence, both outcomes $\set{\set{1,2}}$ and $\set{\set{1},\set{2}}$ remain possible. 
    Therefore, the condition of \Cref{obs:sameSetIsNOM} is satisfied, since the set of possible outcomes, and thus the set of possible coalitions containing agent $1$, stays unchanged when varying $d_{12}$. The same argument can be applied to agent $2$. Hence, the optimal mechanism is NOM.
\end{example}

Next, we will show that this result extends to any number of agents when the weights are arbitrary; however, it no longer holds when the weights are bounded.

\begin{theorem}\label{thm:arbitraryASHGsOPTisNOM}
    For both ASHGs and FHGs with arbitrary weights, any optimal mechanism is NOM.
\end{theorem}
\begin{proof}
Let us fix an agent $i$, and let $C\in \agents_i$. We next show that for any $d_i\in\declarations_i$ there exists $\declared\noti\in \declarations\noti$ such that, in any optimal partition $\pi$ of instance $\declared=(d_i, \declared\noti)$, it holds that $\pi(i)=C$.

Since the weights range in $\mathbb{R}$, for any $d_i$ we can select $\declared\noti\in \declarations\noti$ in such a way that the flattened graph of the instance, $G=(\agents, \hat{w})$, has the weights defined as follows: $\hat{w}(e)= 1$, for all $e=(j,j')$ with $j,j'\in C$, and $\hat{w}(e)= -M$, for a sufficiently large $M$, otherwise. 

Let $\pi$ be any optimal partition for such an instance, and suppose, for the sake of contradiction, that there exists $j \in \pi(i)$  such that $j\not\in C$. Then let us estimate $\Delta_1 = \SW(\pi(i)\setminus\set{j}) + \SW(\set{j}) - \SW(\pi(i)) =  \SW(\pi(i)\setminus\set{j}) - \SW(\pi(i)) $. 

For ASHGs, $\Delta_1 = M\cdot(|\pi(i)| - 1) > 0$.

For FHGs, let us set $s=\sum_{k, k' \in\pi(i)\setminus\set{j}} w_k(k')$, then,
\begin{align*}
    \Delta_1 &= \frac{s}{|\pi(i)| - 1} - \frac{s - M\left(|\pi(i)| - 1\right)}{|\pi(i)|} \\
    &= \frac{s}{|\pi(i)|\cdot\left(|\pi(i)| - 1\right)} + \frac{M\left(|\pi(i)| - 1\right)}{|\pi(i)|} \\
    &\geq \frac{1}{|\pi(i)|}\cdot\left(-\frac{M\left(|\pi(i)| - 1\right)\left(|\pi(i)| - 2\right)}{2\cdot\left(|\pi(i)| - 1\right)} + M\left(|\pi(i)| - 1\right)\right) 
    \\ 
    &= \frac{M}{|\pi(i)|}\cdot\frac{|\pi(i)|}{2} = \frac{M}{2} > 0
\end{align*}
Therefore, it would be preferable to move $j$ to a singleton coalition, which contradicts the optimality of $\pi$. Thus, $\pi(i)\subseteq C$.

Assume there exist two disjoint coalitions $C'$ and $C''$ such that $C', C'' \in \pi$, and $C' \cup C'' \subseteq C$. Since $\pi$ is optimum, $\Delta_2 = \SW(C'\cup C'')- \SW(C') - \SW(C'') \leq 0$, as it cannot be strictly more convenient for the $\SW$ to merge these two coalitions. 
However, in the case of ASHGs, $\Delta_2$ measures the total weight of the cut $(C',C'')$. Since this cut is made of positive edges only, we have $\Delta_2>0$, which is a contradiction to the optimality of $\pi$.
In FHGs, $ \Delta_2= \left( \modulus{C'\cup C''} - 1\right) - (\modulus{C'}-1) - (\modulus{C''}-1)>0$, again a contradiction to the optimality of $\pi$.

In conclusion, for any optimal $\mech$, $Coal_i^\mech(d_i) =\agents_i$, for each $d_i\in\declarations_i$, including the truthful declaration $w_i$ of $i$. Therefore, the condition of \Cref{obs:sameSetIsNOM} is satisfied by any optimal mechanism, and our claim follows.
\end{proof}

\Cref{thm:arbitraryASHGsOPTisNOM} poses NOM in sharp contrast with SP. Unfortunately, such a result does not hold for any class of weights, as the following example shows.

\begin{example}\label{ex:notEveryOPTisNOM}
    Assume weights are bounded. The following arguments apply to both ASHGs and FHGs.
    
    Let $\mech$ be an optimal mechanism that, whenever there are two agents and $d_{12}+ d_{21}=0$, works as follows: if $d_{12}=-1$, the mechanism puts the two agents in separate coalitions; otherwise, the two agents are put in the same coalition. Note that if $d_{12}+ d_{21}=0$, any partition of the two agents is optimal.
    
    If the true type $w_1$ of agent $1$ is such that the score towards agent $2$ is $-x$, with $x<1$, then, $Out_1^\mech(w_1)= \set{\set{\agents}, \set{\set{1}, \set{2}} }$. In fact, whenever $d_{21} >x$ (which is possible as long as $x<1$), the only optimum is achieved by putting the agents in the same coalition; in turn, if $d_{21}<x$, the optimum must split the agents into separate coalitions. 
    
    Assume agent $1$ misreports by declaring $d_1$ where $d_{12}=-1$. In this case, $Out_1^\mech(d_1)= \set{\set{\set{1}, \set{2}}}$. In fact, $d_{12}+ d_{21}\leq 0$ as  $d_{21}\leq 1$, and hence, if $d_{12}+ d_{21} <0$, any optimal outcome will put the agents in singletons; otherwise, the mechanism will apply the tie-breaking rule described above, which shows that the only possible partition is $\set{\set{1}, \set{2}}$.
    Therefore, this manipulation improves the worst case for $1$, proving that the optimal mechanism $\mech$ is not NOM.
\end{example}

We notice that \Cref{ex:notEveryOPTisNOM} is not in contrast with \Cref{thm:arbitraryASHGsOPTisNOM}.

\begin{remark}\label{remark:monotonicity}
NOM, differently from SP, is not monotone w.r.t.\ weights domain inclusions. That is, even if a mechanism is NOM for a class of weights $\mathcal{W}$, it is not necessarily NOM when the admissible weights are restricted to a subset $\mathcal{W}' \subset \mathcal{W}$.
\end{remark}

\Cref{remark:monotonicity} is a natural consequence of \Cref{def:NOM}, in fact, if we restrict the domain of possible weights, we may unpredictably change the space of possible outcomes for a fixed declaration of an agent; consequently, there is no guarantee that the NOM property is preserved under such a restriction.

Even though optimal mechanisms may not necessarily be NOM, in the next section, we provide a characterization of BAPX and NOM mechanisms. As a result, this will show that there always exists an optimal and NOM mechanism, both for bounded and non-negative weights.

\section{ BAPX and NOM mechanisms}
In the previous section, we observed that, while all optimal mechanisms are NOM, when weights range in $\mathbb{R}$, this is not necessarily the case when weights are bounded. However, looking back at \Cref{ex:notEveryOPTisNOM}, we highlight that some strange behavior of the mechanism caused its obvious manipulability. In particular, such a mechanism may return different optimal outcomes for two proportional instances (instances having the same flattened graph or proportional flattened graphs). 

We notice that, given any outcome $\pi$, whenever two agents $i,j$ are in the same coalition, both $w_{ij}$ and $w_{ji}$ contribute to the social welfare; conversely, if they are not in the same coalition, neither $w_{ij}$ nor $w_{ji}$ contributes to it.
In other words, the mutual values of a pair of agents contribute to the social welfare only if the agents are in the same coalition.
For this reason, the optimality of an outcome is \emph{scale-independent}, that is, if an outcome is optimum, it remains optimum for a proportional instance. This leads to the following observation. 
\begin{observation}\label{obs:propSameOPT}
    In ASHGs and FHGs, proportional instances have the same optimal outcomes.
\end{observation}

Therefore, since by \Cref{obs:propSameOPT} proportional instances share the same optimal solutions, it looks reasonable that, when designing an optimal mechanism, one should care only about the flattened graph of the instance, rather than the mutual evaluations of the agents.
For all these reasons, we introduce the property of {\em scale–independence} for a mechanism.

\begin{definition}[Scale-Independence (SI)]
  A mechanism $\mech$ is \emph{scale-independent} (SI) if for any pair of proportional inputs, $\declared, \declared'$, it returns the same outcome, that is, $\mech(\declared)=\mech(\declared')$.  
\end{definition}

SI requires the mechanism to be consistent with the flattened graph, and any graph proportional to it, meaning that the outcome should solely depend on the contribution of pairs of agents to the social welfare, rather than the individual contributions of the agents.

\subsection{A Sufficient Condition for NOM}
In \Cref{ex:notEveryOPTisNOM}, the presented optimal mechanism is not NOM and does not even satisfy SI. With the next theorem, we further emphasize the importance of considering SI mechanisms, as it suffices to guarantee NOM.

\begin{lemma}\label{lemma:propertyP}
Given a space of possible declarations $\declarations$, if
\begin{enumerate}[(\textbf{P})]
    \item\label{item:property} for each $i\in \agents$ and for any $d_i, d'_i \in\declarations_i$, and $\declared_{-i}\in \declarations_{-i}$, there exists $\declared'_{-i}\in \declarations_{-i}$ such that $(d_i, \declared_{-i})\sim_{P}(d'_i, \declared'_{-i})$,
\end{enumerate}
then, any SI mechanism $\mech$ is NOM in the domain $\declarations$. 
\end{lemma}
\begin{proof}
Let $\mech$ be any SI mechanism. 

To prove the statement, we apply  \Cref{obs:sameSetIsNOM}. Therefore, it suffices to show that $Out^\mech_i(d_i)=Out^\mech_i(d_i')$, for every pair of declarations $d_i, d'_i \in \declarations_i$. 

Indeed, consider any outcome $\pi \in Out^\mech_i(d_i)$. By definition, there exists $\declared_{-i} \in\declarations_{-i}$ such that $\pi=\mech(d_i,\declared_{-i})$. According to \ref{item:property}, there exists $\declared_{-i}'\in \declarations_{-i}$ such that $(d_i, \declared_{-i})\sim_{P}(d'_i, \declared'_{-i})$. Then, since $\mech$ is SI, $\pi=\mech(d_i,\declared_{-i})= \mech(d'_i,\declared'_{-i})$, and thus $\pi \in Out^\mech_i(d'_i)$. 

Hence, the condition of \Cref{obs:sameSetIsNOM} is satisfied by any SI mechanism whenever $\declarations$ satisfies \ref{item:property}, and NOM follows.
\end{proof}

\begin{theorem}\label{thm:whenSIisNOM}
For both ASHGs and FHGs with arbitrary, non-negative, or bounded weights, $SI\Rightarrow NOM$.
\end{theorem}
\begin{proof}
We show \ref{item:property} holds true for each analyzed class of weights. 

\medskip \noindent{\em Arbitrary weights.} 
Given an agent $i$, two types $d_i$ and $d'_i$, and any declaration of the other agents $\declared_{-i}$, we define $\declared'_{-i}$ so that $(d_i, \declared_{-i})\sim_{P}(d'_i, \declared'_{-i})$. 

For each $j\in\agents\setminus\set{i}$, we set $d'_{ji} = d_{ij} + d_{ji} - d'_{ij}$ and, for each $j, j'\in\agents\setminus\set{i}$, we set $d'_{jj'}=d_{jj'}$. Notice that this weight setting is feasible for arbitrary weights and the two instances are proportional with $\lambda = 1$.

\medskip \noindent{\em Non-negative weights.}  The transformation we used for arbitrary weights may lead   
to negative values, which is not allowed anymore. However, we can first select $\lambda$ sufficiently large so that $\lambda \cdot (d_{ij} + d_{ji}) - d'_{ij} > 0$, for each $j\in\agents\setminus\set{i}$, and then set values in the following way: 
for each $j\in\agents\setminus\set{i}$, we set $d'_{ji} = \lambda \cdot (d_{ij} + d_{ji}) - d'_{ij}$, and, for each $j, j'\in\agents\setminus\set{i}$, we set $d'_{jj'}=\lambda\cdot d_{jj'}$.

\medskip \noindent{\em Bounded weights.} We use the same transformation we exploited for arbitrary weights. If the resulting weights are not in $[-1,1]$, we normalize them in a proportional way.
\end{proof}

\subsection{A Characterization of BAPX and NOM}\label{sec:characterization}
Let us now specifically focus on BAPX mechanisms. Notice that since optimal mechanisms may not be SI (see \Cref{ex:notEveryOPTisNOM}), the same holds true for BAPX mechanisms.

We next show that, for any possible approximation ratio $\rho$, there always exists an SI mechanism with approximation factor exactly $\rho$. To this end, we will exploit some interesting properties deriving from the definition of the equivalence class $\sim_{P}$ over game instances.
In particular, we introduce a poly-time procedure which, for any pair of instances $\instance, \instance'$ such that $\instance \sim_{P} \instance'$, selects the same representative for their equivalence class. We call this procedure \repr\ and it works as follows:
given a game instance $\instance$, it builds the corresponding flattened graph $G\in\mathcal{G}^n$; next, \repr\ transforms $G=(\agents,\hat{w})$ into a proportional flattened graph $G'=(\agents,\hat{w}')$ such that the maximum weight in absolute terms in $G'$ equals $1$. This is clearly always feasible by normalizing the weights by a factor $\max_{i\neq j}{\modulus{\hat{w}(i,j)}}$. Finally, we assume that there is a fixed ordering of the agents, say $1, \dots, n$, and \repr\ returns an instance with agents $\agents$ and weights $w''_i$, for all $i\in\agents$, defined as follows:
\begin{align*}
w''_{ij}=
\begin{cases}
\hat{w}'(i,j) & \text{if} \ i<j \\
0       & \text{ otherwise.}
\end{cases}
\end{align*}
We will call the instance returned by $\repr(\instance)$ as the \emph{representative instance} of $\instance$.

It is not hard to see that $\repr(\instance)\sim_{P} \instance$, and for any pair of proportional instances $\instance, \instance'$ we have $\repr(\instance)=\repr(\instance')$.
Moreover, the procedure \textsc{Repr} works in polynomial time.

\textsc{Repr} turns out to be very helpful in transforming a mechanism into an SI mechanism while maintaining the same approximation guarantee, as the following lemmas show.
\begin{lemma}\label{lemma:compositionReprIsSI}
 For any mechanism $\mech$, the mechanism $\mech'=\mech\circ\repr$\footnote{We denote $f(g(\cdot))$ by $f\circ g(\cdot)$.} is SI.
\end{lemma}
\begin{proof}
    Given any pair instances such that $\instance \sim_{P}\instance'$, $\repr(\instance)=\repr(\instance')$. Therefore, for instances having proportional flattened graphs, $\mech'$ outputs the same outcome. Hence, the mechanism is SI.
\end{proof}

\begin{lemma}\label{lemma:compositionReprPreservesApprox}
    If $\mech$ is a $\rho$-approximating mechanism, then $\mech'=\mech\circ\repr$ is a $\rho$-approximating mechanism as well.
\end{lemma}
\begin{proof}
    Let $\declared, \declared'\in\declarations$ and $\opt, \opt'$ be the values of a welfare-maximizing allocation in $\instance=(\agents, \declared)$ and $\instance'=(\agents, \declared')$, respectively. 
    
   If $\instance'\sim_{P}\instance$, meaning that their flattened graphs are proportional, there exists $\lambda > 0$ s.t.\ $d_{ij} + d_{ji} = \lambda\cdot \left(d'_{ij} + d'_{ji}\right)$, for all $i,j\in \agents$ with $i\neq j$. Therefore, for both ASHGs and FHGs, $\opt =\lambda\cdot \opt'$ and, for any partition $\pi$, $\SW^\instance(\pi)=\lambda\cdot \SW^{\instance'}(\pi)$. Hence, for any $\pi$, $\frac{\opt}{\SW^\instance(\pi)}=\frac{\opt'}{\SW^{\instance'}(\pi)}$.

    Now, given an instance $\instance$, we denote by $\opt^\instance$ the value of any optimal outcome. Since $\instance\sim_{P} \repr(\instance)$, we can deduce $\frac{\opt^\instance}{\SW^\instance(\pi)}=\frac{\opt^{\repr(\instance)}}{\SW^{\repr(\instance)}(\pi)}$, and hence
    \begin{align*}
       r^{\mech'} &=\sup_{\instance}\frac{\opt^\instance}{\SW^\instance(\mech'(\instance))}
        \\
        &=\sup_{\instance}\frac{\opt^\instance}{\SW^\instance(\mech(\repr(\instance)))} \\
        &= \sup_{\instance}\frac{\opt^{\repr(\instance)}}{\SW^{\repr(\instance)} (\mech(\repr(\instance)))} 
        \leq r^\mech = \rho 
    \end{align*}
    where the last inequality holds because $\mech$ is a $\rho$-approximating mechanism.
\end{proof}

\Cref{lemma:compositionReprPreservesApprox,lemma:compositionReprIsSI} together with \Cref{thm:whenSIisNOM} show the following theorem.

\begin{theorem}\label{thm:SIcharacterization}
    For any $\rho\geq 1$, in ASHGs or FHGs with arbitrary, non-negative, or bounded weights, there exists a NOM and $\rho$-approximating mechanism if and only if 
    there exists an SI and $\rho$-approximating mechanism.
\end{theorem}

\begin{proof}
($\Leftarrow$) Let $\mech$ be an SI and $\rho$-approximating mechanism. By \Cref{thm:whenSIisNOM}, $\mech$ is NOM and the thesis follows as $\mech$ is $\rho$-approximating.

\noindent($\Rightarrow$) Let $\mech$ be a NOM and $\rho$-approximating mechanism. From the previous lemmas, $\mech'=\mech\circ\repr$ is the desired SI and $\rho$-approximating mechanism.
\end{proof}

Essentially, Theorem~\ref{thm:SIcharacterization} allows us to focus our attention on generic $\rho$-approximating mechanisms. Indeed, each such algorithm can be turned into a $\rho$-approximating and SI mechanism through the procedure \textsc{Repr}, and hence into a $\rho$-approximating and NOM mechanism, by Theorem~\ref{thm:whenSIisNOM}.

\subsection{Approximate NOM mechanism}
We recall that, in ASHGs, whenever the agents' weights are non-negative, the optimum,  which consists in grouping all the agents in the same coalition, is SP. In turn, no SP mechanism has a bounded approximation ratio for arbitrary and bounded weights. Moreover, regardless of the SP or NOM requirement, no polynomial-time algorithm can provide an approximation in ASHGs that is $O(n^{1-\epsilon})$ for every $\epsilon>0$ unless P = NP even in simple cases (\cite{flammini2022strategyproof}), while computing welfare-maximizing partitions in FHGs is NP-hard~(\cite{aziz2015welfare}). In what follows, we focus on determining NOM and efficient mechanisms for arbitrary, non-negative, and bounded weights. 

\begin{mechanism}\label{mech:matchingASHGs}
Given an instance $\instance$, let $\mech$ be a mechanism that builds the flattened graph corresponding to $\instance$, computes a maximum weighted matching on it, and creates a coalition for each matched pair of agents; unmatched agents are put in singletons. Then, we define $\ref{mech:matchingASHGs}=\mech\circ\repr$.
\end{mechanism}

\begin{theorem}\label{thm:matchingASHGsisNOM}
   \ref{mech:matchingASHGs} is NOM for both ASHGs and FHGs under unbounded, non-negative, and bounded weights. Furthermore, it provides an $n$-approximation of $\opt$ for ASHGs and a $2$-approximation for FHGs.
\end{theorem}

To show the theorem, we will make use of the following well-known fact (see e.g.\ Problem 16.5 of \cite{soifer2009mathematical}):

\begin{fact}\label{fact:cliqueMatchingPartition}
    Let $K$ be a clique of size $k$.
    Its edges can be partitioned into $k-1$, if $k$ is even, and $k$, if odd, disjoint matchings.
\end{fact}
\begin{proof}[Proof of \Cref{thm:matchingASHGsisNOM}]
    (NOM) By \Cref{lemma:compositionReprIsSI}, \ref{mech:matchingASHGs} is SI; therefore, from \Cref{thm:whenSIisNOM}, \ref{mech:matchingASHGs} is NOM for instances with arbitrary, non-negative, or bounded weights. 
    
    (Approximation) Let $\mech$ be the subroutine as described in the definition of \ref{mech:matchingASHGs}. We next show the approximation guarantee of $\mech$  for both ASHGs and FHGs. Since $\ref{mech:matchingASHGs}=\mech\circ\repr$, by \Cref{lemma:compositionReprPreservesApprox}, this approximation guarantee is preserved in \ref{mech:matchingASHGs}. Let $\pi$ denote the outcome of $\mech$.
    
    Let us consider an optimal partition $\OPT$ with value $\opt$, and let $k$ be the number of its coalitions, say, $C_1, \dots, C_k$. For each coalition $C_h$, let $c_h = |C_h|$ be its size, and let $G_h$ be the subgraph of the flattened graph containing only vertices corresponding to agents in $C_h$, with $E_h$ denoting the edges of $G_h$. From \Cref{fact:cliqueMatchingPartition}, we can partition $E_h$ in at most $c_h$ disjoint matchings $M_h^1, \ldots, M_h^{c_h}$. Let instead $M^*_h$ to denote the maximum matching in $G_h$, and let $m^*_h$ be its value. By optimality of $M^*_h$, we have that $\sum_{e \in M_h^j} w_e \leq m^*_h$ for every $j = 1, \ldots, c_h$. Hence, $\sum_{e \in G_h} w_e = \sum_{j = 1}^{c_h} \sum_{e \in M_h^j} w_e \leq c_h m^*_h \leq n m^*_h$.

    Now let us consider the maximum matching $M^*$ computed by $\mech$, and let $m^*$ denote its total value (the sum of the weights of its undirected edges). Since $M = \cup_{h=1}^k M_h $ is a matching in the flattened graph of value $\sum_{h = 1}^k m^*_h$, by optimality of $M^*$ we have that $m^* \geq \sum_{h = 1}^k m^*_h$.

    In the case of ASHGs, we have that $\SW(C_h)=\sum_{e\in G_h} w_{e}$ and $SW(\pi) = \sum_{e \in M^*} w_e = m^*$. Thus,
    \begin{align*}
     \opt = \sum_{h=1}^k \SW(C_h) \leq \sum_{h=1}^k n\cdot m^*_h \leq n\cdot m^* = n\cdot \SW(\pi) \, ,   
    \end{align*}
    and the approximation follows.
     
    In FHGs, $SW(C_h) = \frac{\sum_{e \in G_h} w_{e}}{c_h} \leq m^*_h$ and $SW(\pi) = \sum_{e \in M^*} \frac{w_e}{2} = \frac{m^*}{2}$, and hence 
    \begin{align*}
      \opt = \sum_{h=1}^k \SW(C_h) \leq \sum_{h=1}^k m^*_h \leq m^* = 2\cdot \SW(\pi) \, ,  
    \end{align*}
    and the approximation follows.
\end{proof}

\section{NOM and Optimality in Discrete Settings}
So far we have seen that, when weights range in $\mathbb{R}$, $\mathbb{R}_{\geq0}$, or $[-1,1]$, optimality and NOM are compatible. However, in \cite{flammini2025non} it has been shown that for ASHGs with weights in $\set{-n, 1}$, no optimal mechanism is NOM, while if the weights are in $\set{-\frac{1}{n},1}$, there exists a NOM and optimal mechanism. 
Therefore, even in simple discrete settings, the compatibility of NOM with optimality depends on the specific values the weights may take.

In what follows, we will dive into these known results to provide a more complete understanding of the compatibility of NOM with optimality in discrete settings. Since existing results focus only on ASHGs, we consider this class as well and assume general duplex valuations, i.e., weights that may take values in $\set{-x,0, 1}$ for $x\in (0, \infty)$, which contains and extends the setting considered in \cite{flammini2025non}. 

We are able to prove the following characterization:
\begin{theorem}\label{thm:generalDuplex}
Let $\parity=1$, if $n$ is even, and $\parity=0$, otherwise.
For general duplex valuations, the following statements hold true:

\begin{itemize}
    \item for $x \geq 2n -3$, there exists an optimal and NOM mechanism;
    \item for $x \in (1, 2n-3)$, no optimal mechanism is NOM;
    \item for $x \in \left[\frac{1}{n+\parity}, 1\right]$, there exists an optimal and NOM mechanism; 
    \item for $x \in \left(\frac{2}{n^2 - 3 +\parity}, \frac{1}{n+\sqrt{2n + 2 - \parity}}\right)$, no optimal mechanism is NOM;
    \item for $x \in \left(0, \frac{2}{n^2 - 3 +\parity}\right]$,  there exists an optimal and strategyproof (and hence NOM) mechanism.
\end{itemize}
\end{theorem}

Firstly, we introduce some preliminary definitions and results that will be useful later, and next we prove \Cref{thm:generalDuplex} case by case in distinct subsections.

\subsection{Preliminary Results}
For ease of explanation, most of the proofs of this section are postponed to~\ref{sec:appendix}. 

\begin{lemma}\label{lemma:bestcase}
    For any $x \in \mathbb{R}^+$ and any agent $i \in \agents$,
    there exists a game instance where, when declaring truthfully, in the unique optimum $i$ gets her maximum possible utility.
\end{lemma}
\begin{proof}
    Let us consider the game instance where $i$ declares truthfully and $\declared_{-i}$ is the following:
    \begin{itemize}
        \item for $j \in \agents\setminus\set{i}$, if $w_i(j) = 1$, then $w_j(i) = 1$;
        \item for $j, k \in \agents\setminus\set{i}$, if $w_i(j) = w_i(k) = 1$, then $w_j(k) = w_k(j) = 1$;
        \item all other weights are set to $-x$.
    \end{itemize}
    Here, regardless of $x$, in the unique optimum, every agent $j \in \agents\setminus\set{i}$ for whom $w_{ij} = 1$ is put into the same coalition as $i$, and all other agents are in the singletons, which is the best possible outcome for $i$.
\end{proof}

From \Cref{lemma:bestcase}, it directly follows that Condition~1 of \Cref{def:NOM} holds for every optimal mechanism. In fact, if by declaring truthfully the highest possible utility may be achieved, no other declaration may possibly guarantee a higher utility.
Therefore, non-obvious manipulability depends only on whether Condition~2 is satisfied as well, i.e., on verifying if the worst outcome that can be achieved by being truthful is at least as good as the worst that can be achieved by misreporting.

The following two definitions and lemmas about their properties were previously introduced in \cite{flammini2025non} for the case where the agents valuations belong to the set $\set{-\frac{1}{n}, 1}$. Here, we generalize them for the case of duplex valuations with $x \leq 1$.
\begin{definition}[Almost Isolated Clique]
    $C \subseteq \agents\setminus\set{i}$ is an almost isolated clique with the hinge node $i$ if:
    \begin{itemize}
        \item for every $j, k \in C$, $w_{jk} = w_{kj} = 1$;
        \item for every $j \in C$ and $k \in \agents\setminus\left(C \cup \set{i}\right)$, $w_{jk} = w_{kj} = -x$.
    \end{itemize}
\end{definition}

\begin{lemma}\label{lemma:AlmostIsolated}
    If $x < 1$ and $C \subseteq \agents \setminus\set{i}$ is an almost isolated clique with the hinge node $i$, then in any optimal $\pi^*$ there exists $C' \in \pi^*$ such that $C \subseteq C'$. Furthermore, if $i\notin C'$, then $C = C'$. If, instead, $x = 1$, there always exists an optimal partition $\pi^{**}$ such that, for some $C'' \in \pi^{**}$, $C \subseteq C''$.
\end{lemma}

\begin{figure}[tb]
	\centering
    \def \scale {1.66}
	\def \variable {0.3cm}
    \def \var {0.5cm}
	\pgfdeclarelayer{bg}    
	\pgfsetlayers{bg,main} 

\centering
	\begin{tikzpicture}[-,>=stealth',shorten >=0pt,auto,node distance=\variable,
	thick,main node/.style={circle,draw,fill, minimum size=1pt, scale=0.3}]

	\foreach \numb in {1,2,3,4,5} {
		\node[main node] (\numb) at ({72*\numb +55}:\var*\scale) {};
	}
 \foreach \numb in {6,7,8,9,10} {
    \node[main node] (\numb) [yshift= -7cm, color=white] at ({72*\numb +108}:\variable*\scale) {};
	}
    
	\node[main node] (i) [below of = 3, yshift= -1.3cm *\scale , color=white, minimum size=35pt] {};
    \node[main node] (i2) [below of = 3, yshift= -1.3cm *\scale,  circle,draw,fill, minimum size=13pt]  {};
	\node () [ right of =i2, font=\footnotesize] {$i$};
    \node[main node] (a) [below of = 3, yshift= -5.5cm*\scale] {};
    \node[main node] (b) [below left of = 3, yshift= -5cm*\scale, xshift= -3.5 cm *\scale] {};
    \node[] () [right of = a, xshift=0.2cm*\scale] {$\cdots$};
    \node[main node] (c) [below right of = 3, yshift= -5cm*\scale, xshift= 3.5cm*\scale ] {};
    \node[] () [left of = a, xshift=-0.2cm*\scale] {$\cdots$};
    \node[] () [ right of =4] {$H$};

 \path[every node/.style={font=\sffamily\small}, edge/.style ={dashed}]

    \foreach \B in {1,2,3,4,5}{   
		
		(\B) edge [line width=0.3pt, ->] node {} (i)
	}
	\foreach \A in {1,2,3,4,5}{
		\foreach \B in {1,2,3,4,5}{   
			\ifnum \A<\B
			(\A) edge [line width=0.3pt] node {} (\B)
			\fi
		}
	}
 (i) edge [line width=0.3pt, ->, bend right, dashed] node {} (4)
(i) edge [line width=0.3pt, ->, bend left] node {} (2)
(i) edge [line width=0.3pt, ->, dashed] node {} (b)
(i) edge [line width=0.3pt, ->] node {} (a)

	;
\end{tikzpicture}

\caption{Octopus graph having center $i$. Undirected edges represent mutually equal weights. Solid edges represent positive weights. Dashed edges represent weights equal to zero. Absence of an edge implies negative weight.}\label{fig::octopus}
\end{figure}

\begin{definition}[Octopus Graph]
Given an agent $i$ and $H\subseteq \agents\setminus\set{i}$, a game instance is an $i$-centered {\em octopus graph} with the head $H$ if:
\begin{itemize}
    \item for every $j, k \in H$, $w_{jk} = 1$;
    \item for each $j \in H$, $w_{ji} = 1$;
    \item for each $j \in  \agents \setminus\set{i}$ and $k \in \agents\setminus\left(\{i\} \cup H\right)$, $w_{jk} = w_{kj} = w_{ki} = -x$;
    %\item for each $j \in \agents \setminus\set{i}$, ...
\end{itemize}
\end{definition}
An example of an octopus graph is depicted on \Cref{fig::octopus}.

\begin{observation}\label{obs:octopus}
    If a game instance is at the same time an $i$-centered octopus graph with the head $H$ and an $i'$-centered octopus graph with the head $H'$ for some $i \neq i'$, then $H \cup \set{i} = H' \cup i'$ is a bidirectional clique in the graph induced by the positive weights, and all other weights are negative.
\end{observation}

\begin{lemma}\label{lemma:duplexNoNegativeRelation}
If $x \geq 2n -3$, there always exists an optimal outcome $\pi^*$ such that, for every coalition $C \in \pi^*$, either $C$ is a singleton, or for every $i, j \in C$ it holds that $w_{ij}, w_{ji} \geq 0$.   
\end{lemma}
\begin{lemma}\label{lemma:OctopusOPT}
    If $x \leq 1$ and a game instance is an $i$-centered octopus graph with the head $H$ such that $|H| \geq \max\left\{\frac{1-x}{2x}; 1\right\}$, then there exists an optimum partition $\pi^*$ such that $H \cup \set{i} \in \pi^*$.
\end{lemma}

\subsection{A NOM and Optimal Mechanism}\label{sec:NOMmechDuplex}

Notice that \Cref{thm:generalDuplex} purely addresses the existence of a mechanism combining non-obvious manipulability and optimality, while the computational complexity is out of  discussion. Furthermore, the existence of an optimal NOM mechanism does not imply that all optimal mechanisms are NOM, which was already shown for the case $x=-1/n$~(\cite{flammini2025non}). Therefore, in the~mechanism we present next, we will have to specify certain tie-breaking rules to guarantee non-obvious manipulability.

\begin{mechanism}\label{mech:duplex}
   The mechanism returns an optimal partition $\pi^*$; if more than one optimal outcome exists, it breaks ties according to the rule $\textsc{TieBreaking}(x)$.
\end{mechanism}
$\textsc{TieBreaking}(x)$ is defined for the specific values of $x$ in the corresponding subsections. 

Throughout the next subsections, we prove~\Cref{thm:generalDuplex}. Recall that, by \Cref{lemma:bestcase}, for any optimal mechanism, the best-case scenario can never be improved, and therefore it is sufficient to consider the~worst-case scenario.

\subsection{Case $x\geq 2n -3$: There exists an optimal NOM mechanism}
Let us start by defining the tie-breaking rule for this case.

If $x \geq 2n - 3$, $\textsc{TieBreaking}(x)$ chooses an optimal outcome $\pi^*$ such that there are no edges with negative weights within coalitions: for each $C \in \pi^*$ and every $i, j \in C$ it holds that $w_{ij}, w_{ji} \geq 0$. 
Note that, by \Cref{lemma:duplexNoNegativeRelation}, this is always possible.

We now show that in this case \ref{mech:duplex} is NOM. Recall that, since \ref{mech:duplex} is optimal, by \Cref{lemma:bestcase}, we need to consider only Condition 2 of \Cref{def:NOM}, that is, the worst-case scenario.

If $i$ truthfully reports $w_i$, our choice of a tie-breaking rule guarantees that no outcome output by \ref{mech:duplex} can give her negative utility as no agent from $\delta^-(w_i)$ can be put in the same coalition as her. Now, let us consider an arbitrary $d_i$ and show that there exists $\declared_{-i}$ such that in the outcome of \ref{mech:duplex} $i$ ends up in a singleton coalition. Indeed, our tie-breaking rule guarantees that it suffices to set $d_{ji}= -x$ for each $j\neq i$. In conclusion, the worst-case scenario for $i$ when truthfully reporting is not worse than being assigned to the singleton coalition, while the worst-case scenario when misreporting is at least as bad as ending up in the singleton. 
Therefore, Condition~2 of \Cref{def:NOM} is satisfied.

\subsection{Case $x \in (1, 2n-3)$: No optimal mechanism is NOM}
First, observe that in order for this interval to be non-empty, we require $n \geq 3$.

\begin{lemma}\label{lemma:singleton}
    Let $d^*_i$ be the declaration where $d^*_{ij}= -x$ for all $j\in\agents \setminus \set{i}$. If $x>1$, then, for any optimal mechanism $\mech$ and any $\pi \in Out_i^\mech(d^*_i)$, it holds that $\pi(i)=\set{i}$.
\end{lemma}
\begin{proof}
    For the sake of contradiction, assume that there exists $\pi' \in Out_i^\mech(d^*_i)$, which is optimal for some $(d^*_i, \declared_{-i})$ and for which $\modulus{\pi'(i)}>1$. Let $C= \pi'(i)$ and $c = |C| > 1$. Since, for each agent $j \in C\setminus\set{i}$, $w_{ji} \leq 1$, the total value of the cut $(C\setminus\set{i}, \set{i})$ in the instance $(d^*_i, \declared_{-i})$ is at most $c-1 - (c-1)x =  (c-1)(1-x) $, and therefore is always negative as $x>1$. Let $\pi''$ be the partition derived from $\pi'$ by moving $i$ to the singleton coalition. $\SW(\pi') - \SW(\pi'') = w(\cut(C\setminus\set{i}, \set{i})) < 0$, which contradicts the optimality of $\pi'$ for the instance $(d^*_i, \declared_{-i})$.
\end{proof}

\begin{lemma}\label{lemma:grandoralmost}
 Suppose that $1 < x < 2n-3$. Let us fix $i \in \agents$, and let $\declared_{-i}$ be the declarations of the other agents where $d_{jj'}=1$ for every $j\in \agents\setminus\set{i}$ and $j'\in \agents$. Then, regardless of $d_i$, the outcome of an~optimal mechanism for the instance $(d_i, \declared_{-i})$ is either $\pi= \set{\agents}$ or $\pi'=\set{\agents\setminus\set{i}, \set{i}}$.
\end{lemma}
\begin{lemma}\label{lem:kforxneqn-2}
Suppose that $1 < x < 2n-3$ and $x\neq n-2$.
Then there exists an integer $k_x < n-1$ such that 
\begin{equation}\label{eq:conditionPoitiveCut}
2(n-1) - (1+x)k_x > 0,
\end{equation}
and 
\begin{equation}\label{eq:conditionNegativeUtility}
    n-1 - (1+x)k_x <0.
\end{equation}
\end{lemma}

Next, we further distinguish two cases based on the value of $x$.

Suppose first that $x \neq n-2$, which will allow us to use \Cref{lem:kforxneqn-2}.
Then, let us consider a~truthful declaration of $i$, $w_i\neq d^*_i$, such that $\modulus{\delta^-(w_i)} = k_x$, where $k_x$ is as defined in \Cref{lem:kforxneqn-2}. 
Note that, if $\declared_{-i}$ is such that $d_j(j')=1$ for each $j \in \agents\setminus\set{i}$ and $j'\in \agents$, then the only optimum puts the agents in the grand coalition. Indeed, by \Cref{lemma:grandoralmost}, the optimum can only be either $\pi= \set{\agents}$ or $\pi'=\set{\agents\setminus\set{i}, \set{i}}$, but
$$
\SW(\pi) - \SW(\pi') = w(\cut(\agents\setminus\set{i}, \set{i}) = 2(n-1) - (1+x)k_x >0
$$
because of the way we have chosen $k_x$, which establishes that $\pi$ is the unique optimum in this setting, and the utility of $i$ is $u_i(\agents) = n-1 - (1+x)k_x <0$.
As a consequence, in the worst outcome in $Out^\mech(w_i)$, the utility of $i$ is negative. 
Therefore, declaring $d^*_i$ such that $d^*_{ij} = -x$ for each $j \in \agents\setminus\set{i}$ rather than $w_i$ would improve the worst-case scenario for $i$ since, from \Cref{lemma:singleton}, she will always be put in the singleton coalition. This concludes the case $x \neq n-2$.

It remains to consider the case that $x=n-2$.
Let us pick three arbitrary agents $i, j^*, j^{**} \in \agents$, and let us consider the following truthful declaration of $i$: $w_{ij} = 1$ for each $j \in \agents\setminus\set{i, j^{*}, j^{**}}$, and $w_{ij} = -x$ for $j \in \set{j^{*}, j^{**}}$.

Let us consider $\declared_{-i}$ such that $d_{jj'} = 1$ for every $j \in \agents\setminus\set{i, j^{**}}$ and $j' \in \agents\setminus\set{j^{**}}$, while all other weights are set to $-x$. Then, in the instance $(w_i, \declared_{-i})$ in any optimal partition $j^{**}$ will be put into the singleton coalition since there are only negative relations between $j^{**}$ and $\agents\setminus\set{j^{**}}$, and, therefore, for any $C \in \agents_{j^**}$ with $|C| \geq 2$, replacing $C$ with $C \setminus\set{j^{**}}$ and $\set{j^{**}}$ increases the social welfare of the partition by $2x(|C| - 1) > 0$.

Thus, in any optimum, the agents from $\agents\setminus\set{j^{**}}$ will be partitioned in the same way they would be in some optimum of the instance with $\agents' = \agents\setminus\set{j^{**}}$, $n' = n - 1$, and $x = n' - 1 = n - 2$ and the declarations of the agents $(w'_i, \declared'_{-i})$ preserved from $(w_i, \declared_{-i})$. Since $\declared'_{-i}$ and $x$ satisfy the conditions of \Cref{lemma:grandoralmost}, in the original instance, there are only two possible optimal partitions: $\pi_1 = \set{\agents\setminus\set{i, j^{**}}, \set{i}, \set{j^{**}}}$ and $\pi_2 = \set{\agents\setminus\set{j^{**}}, \set{j^{**}}}$. However, 
\[
\SW(\pi_2) - \SW(\pi_1) = 2(n - 2) - 1 - x = n - 3 > 0,
\]
which makes $\pi_2$ the unique optimum. Therefore, the worst-case utility of $i$ is at least as bad as $n - 3 - x = - 1 < 0$.

On the other hand, if $i$ declares $d^*_{ij} = - x$ for each $j \in \agents\setminus\set{i}$, then, from \Cref{lemma:singleton}, $i$ will be in the singleton coalition regardless of $\declared_{-i}$. Thus, such manipulation will improve the worst-case utility of $i$.

\subsection{Case $x \in \left[\frac{1}{3}, 1\right]$: There exists an optimal and NOM mechanism}

For $x \in \left[\frac{1}{3}, 1\right]$, we set $\textsc{TieBreaking}(x)$ so that, if a game instance is an $i$-centered octopus graph with the head $H$ and there exists an optimum $\pi^*$ where $H \cup \set{i} \in \pi^*$, \ref{mech:duplex} outputs $\pi^*$. Furthermore, if a game instance is not an octopus graph, \ref{mech:duplex} outputs an optimal partition with the smallest number of coalitions.

Note that from \Cref{obs:octopus} it follows that the mechanism is well-defined: if for a given instance $i$ and $H$ are not uniquely defined, then there is only one optimum consisting of $H \cup \set{i}$ as one of the coalitions and all other agents put in singletons. 
Here, we distinguish between two cases. 
First, let us suppose that $\agents\setminus\left(\delta^+(i)\cup\set{i}\right) \neq \emptyset$, i.e., $i$ has at least one non-positively weighted edge towards other agents.
    In this case, let us consider the $i$-centered octopus graph with the head $H = \agents\setminus\left(\delta^+(i)\cup\set{i}\right)$. For $x \geq \frac{1}{3}$, $\frac{1-x}{2x} = \frac{1}{2x} - \frac{1}{2} \leq \frac{3}{2} - \frac{1}{2} = 1$, and, therefore, $\max\left\{\frac{1-x}{2x}; 1\right\} = 1$. Since $|H| \geq 1 = \max\left\{\frac{1-x}{2x}; 1\right\}$, from \Cref{lemma:OctopusOPT} and $\textsc{TieBreaking}(x)$, it follows that \ref{mech:duplex} returns a partition such that $i$ is in the coalition $H \cup \set{i}$ regardless of $d_i$. Since this results in $i$ getting the worst possible utility, no manipulation can improve her worst case.

    Now, let us consider the case when $\agents\setminus\left(\delta^+(i)\cup\set{i}\right) = \emptyset$, i.e., $w_{ij} = 1$ for every $j \in \agents \setminus\set{i}$. First, let us note that there is no such instance where $i$ ends up in the singleton when truthfully reporting, because, for any coalition $C \subseteq \agents\setminus\set{i}$, merging $\set{i}$ and $C$ into a single coalition $C\cup\set{i}$ results into the change in the social welfare at least equal to $|C|(1-x) \geq 0$, where zero means a~tie that \ref{mech:duplex} will always resolve in favor of merging since this strictly decreases the number of coalitions in the partition. Therefore, the worst-case scenario for $i$ is being in the coalition with only one friend.
    
    Let us pick an arbitrary $j \in \agents\setminus\set{i}$ and set $\declared_{-i}$ in such a way that the instance is an $i$-centered octopus graph with the head $H = \set{j}$. For this instance, by \Cref{lemma:OctopusOPT} and $\textsc{TieBreaking}(x)$, \ref{mech:duplex} will always return $\set{i, j}$ as one of the coalitions regardless of $d_i$. Thus, for every declaration, the worst-case utility of $i$ is either as bad as when reporting truthfully or worse.

\subsection{Case $x \in \left[\frac{1}{n+\phi_n},\frac{1}{3}\right)$: there exists an optimal and NOM mechanism}
For $x \in \left[\frac{1}{n + \parity}, \frac{1}{3}\right)$, we set $\textsc{TieBreaking}(x)$ so that, if a game instance is an $i$-centered octopus graph with the head $H$ and there exists an optimum $\pi^*$ such that $H \cup \set{i} \in \pi^*$, \ref{mech:duplex} outputs $\pi^*$.

We show that mechanism \ref{mech:duplex} is NOM for these values of $x$ as well. In this case, whether for a~given agent $i$ with truthful declaration $w_i$ there exists an instance in which \ref{mech:duplex} outputs a~partition giving $i$ their worst possible utility depends on interplay between $x$ and $w_i$, and it is no longer sufficient for $i$ to have at least one non-positive weight in $w_i$.

For a fixed agent $i \in \agents$ let us define the sets of enemies $E_i = \delta^-(i) = \set{j \,\vert\, w_{ij} = -x}$, friends $F_i = \delta^+(i) = \set{j \,\vert\, w_{ij} = 1}$, and neutrals $N_i = \set{j \,\vert\, w_{ij} = 0}$ respectively. Note that $E_i \cup N_i \cup F_i \cup \set{i} = \agents$.
We distinguish two cases based on the value of $|E_i| + |N_i|$.

If $|E_i| + |N_i| \geq \frac{1-x}{2x}$, let us consider an $i$-centered octopus graph with the head $H = E_i \cup N_i$. From \Cref{lemma:OctopusOPT}, \ref{mech:duplex} returns an optimal partition $\pi^*$ such that $H \cup \set{i} \in \pi^*$, regardless of the preferences $d_i$ submitted by $i$. Since this provides the lowest possible utility for agent $i$, this is sufficient to prove that there is no profitable deviation in this case.

If, instead, $|E_i| + |N_i| < \frac{1-x}{2x}$, let us consider an $i$-centered octopus graph with the head $H = E_i \cup N_i \cup F'_i$, where $F'_i \subseteq F_i$ such that $|F'_i| = \left\lceil\frac{1-x}{2x}\right\rceil - |E_i| - |N_i|$. For this instance, $|H| = \left\lceil\frac{1-x}{2x}\right\rceil$, so, according to \Cref{lemma:OctopusOPT} and $\textsc{TieBreaking}(x)$, \ref{mech:duplex} returns an optimal $\pi^*$ such that $H \cup \set{i} \in \pi^*$, and thus $i$ gets utility $|F'_i|-x|E_i|$, regardless of her declaration.
\begin{lemma}
\label{lem:final1n13}
 For every possible instance in which $i$ submits her truthful declaration, \ref{mech:duplex} guarantees to $i$ a utility at least equal to $|F'_i|-x|E_i|$.
\end{lemma}
Non-obvious manipulability of \ref{mech:duplex} follows since \Cref{lem:final1n13} shows that getting the utility $|F'_i|-x|E_i|$ is indeed the worst-case scenario for $i$ when reporting truthfully, and this scenario, as already proven, cannot be improved by a false declaration.

\subsection{Case $x \in \left(\frac{2}{n^2 - 3 + \parity}, \frac{1}{n+\sqrt{2n + 2 - \parity}}\right)$: no optimal mechanism is NOM}
\label{subsec:case_splitted}
Note that for this interval to be distinct from the intervals above, we need $n \geq 4$.

Here, we will present only the proof for the case when $\frac{2}{n^2 - 3 + \parity} < x < \frac{1}{2n - 3}$. The proof for the case when $\frac{1}{2n - 3} \leq x < \frac{1}{n+\sqrt{2n + 2 - \parity}}$ follows a very similar approach and thus is deferred to the appendix.
We start with the following observation:
\begin{lemma}\label{lemma:WeaklyConnectingAgent}
    If $x < \frac{1}{2n - 3}$ and there exists $i \in \agents$ such that $w_{ij} = 1$ for every $j \in \agents\setminus\set{i}$, the grand coalition is the unique optimum.
\end{lemma}
\begin{proof}
    If $x < \frac{1}{2n - 3}$, $\frac{1 + x}{2x} > \frac{2n - 3}{2} + \frac{1}{2} = n - 1$, and then the statement of the lemma is equivalent to Lemma~\ref{lemma:all_ones_manipulation} from \Cref{sec:appendix}.
\end{proof}

We next show that \Cref{lemma:WeaklyConnectingAgent} allows agent $i$ to find an obvious manipulation for any optimal mechanism, and thus prove that there is no optimal and NOM mechanism. Let us consider the~following example:
 \begin{example}\label{ex:no_opt_nom}
    Let us consider an instance such that the agent $1$ has the following weights: $w_{12} = 1$ and $w_{1i} = -x$ for all $i \geq 3$. Let us then consider $\declared_{-1}$ where all agents from $\agents\setminus\set{1}$ are split into two groups, $C_1$ and $C_2$, in such a way that $|C_1| = \left\lceil\frac{n}{2}\right\rceil - 1$, and $|C_2| = \left\lfloor\frac{n}{2}\right\rfloor$, and $2 \in C_2$, and the~weights are the following:
    \begin{itemize}
        \item for every $i, j \in C_1$, $w_{ij} = w_{ji} = w_{i1} = 1$;
        \item for every $i, j \in C_2$, $w_{ij} = w_{ji} = 1$;
        \item all other weights in $\declared_{-i}$ are set to $-x$.    \end{itemize}
    \end{example}

    We next prove that, in \Cref{ex:no_opt_nom}, $\pi^* = \set{C_1 \cup \set{1}, C_2}$ is the unique optimal partition. Since both $C_1$ and $C_2$ are almost-isolated cliques, by \Cref{lemma:AlmostIsolated}, in any optimum, all agents from $C_1$ are in the same coalition, and the same holds for the agents from $C_2$. Hence, apart from $\pi^*$, there are only three options for an optimal partition, all of which we will next show to be worse than $\pi^*$.
    \begin{description}
        \item[$\pi_1 = \set{C_1, C_2 \cup\set{1}}$:]
        \begin{align*}
            \SW(\pi_1) - \SW(\pi^*) & = 1 - x(2|C_2| - 1) + x|C_1| - |C_1| \\
            &= 1 + x - 2x \left\lfloor\frac{n}{2}\right\rfloor - (1 - x)\left( \left\lceil\frac{n}{2}\right\rceil - 1\right)
        \end{align*}
        If $n$ is even, then
        \begin{align*}
         \SW(\pi_1) - \SW(\pi^*) & = 2 - nx - (1 - x)\cdot\frac{n}{2} = 2 - \frac{n}{2} - x\cdot\frac{n}{2} < 0,
        \end{align*}
        where we used that $n \geq 4$.
        
        If $n$ is odd, then
        \begin{align*}
         \SW(\pi_1) - \SW(\pi^*) & = 1 + x - x(n-1) - (1 - x)\cdot\frac{n-1}{2} = 1 - \frac{n-1}{2} - x\cdot \frac{n-3}{2} < 0,
        \end{align*}
        where we used that $n > 3$.
        \item[$\pi_2 = \set{C_1 \cup C_2 \cup \set{1}}$:]
        \begin{align*}
            \SW(\pi_2) - \SW(\pi^*) & = -2x|C_1||C_2| + 1 - x(2|C_2| - 1)
            = -2x|C_2|(|C_1| + 1) + 1 + x
        \end{align*}
        If $n$ is even, then 
        \begin{align*}
            \SW(\pi_2) - \SW(\pi^*) & = 1 - x\cdot\frac{n^2 - 2}{2} < 1 - \frac{2}{n^2 - 2}\cdot \frac{n^2 - 2}{2} = 0.
        \end{align*}
        If $n$ is odd:
        \begin{align*}
            \SW(\pi_2) - \SW(\pi^*) & = 1 - x\cdot\frac{n^2 - 3}{2} + 1 < 1 -\frac{2}{n^2 - 3}\cdot \frac{n^2 - 3}{2} = 0.
        \end{align*}
        \item[$\pi_3 = \set{C_1, C_2, \set{1}}$:]
        \begin{align*}
            \SW(\pi_3) - \SW(\pi^*) = -|C_1| + x |C_1| = -(1 - x)|C_1| < 0.
        \end{align*}
    \end{description}

    So, $\pi^*$ is the unique optimum for this instance, in which the utility for $1$ is $-x|C_1| < 0$. Hence, in the worst-case utility of $1$ is at most equal to $-x|C_1| < 0$.

    However, if $1$ declares $d_{1i} = 1$ for every $i \in \agents \setminus \set{1}$, then, from \Cref{lemma:WeaklyConnectingAgent}, for any $\declared_{-i}$, the~grand coalition is the unique optimum, which means that $1$ always gets the utility equal to $1 - x(n - 2) > 1 - \frac{n-2}{2n-3} > 0$. Thus, $1$ can improve her worst case by misreporting, and an optimal mechanism cannot be NOM.

\subsection{Case $x \in \left(0; \frac{2}{n^2 - 3 + \parity}\right]$: there exists an optimal and strategyproof mechanism}

If $0 < x \leq \frac{2}{n^2 - 3 + \parity}$, $\textsc{TieBreaking}(x)$ outputs an optimal partition where the number of edges with positive weights between the~coalitions is minimized, and, in case of further ties, it chooses a partition maximizing the number of coalitions. Any other tie is broken arbitrarily.

\begin{lemma}\label{lem:weakly_connected_sp}
    If $x \leq \frac{2}{n^2 - 3 + \parity}$, then $\ref{mech:duplex}$ outputs the weakly connected components of the directed graph induced by the positive weights between agents.
\end{lemma}

Finally, we show that \ref{mech:duplex} is strategyproof for these values of $x$. A similar result has already been proven by \cite{flammini2022strategyproof} for the case when the agents' declarations belong to the set $\set{1, -\frac{1}{n}}$, and can be shown in our case as well using the same reasoning. Let us now provide a~brief informal summary of the main steps of the proof. Firstly, let us note that an agent $i \in \agents$ never benefits from including in her coalition agents from other weakly connected components (of the directed graph induced by the positive weights between agents) since she does not have a~positive weight for any of them. Secondly, the only way for $i$ to separate herself from some of the agents she dislikes that are in her own weakly connected component $C$ requires that $i$ separates herself from at least one of the agents for whom she has a positive weight as well, which always decreases $i$'s utility as $1 > xn$ and therefore losing even one positively valuated agent from $C$ can never be counterbalanced by losing any feasible number of negatively valuated agents. Thus, misreporting never increases $i$'s utility and \ref{mech:duplex} is strategyproof. 

\section{Conclusions and Future Work}
The goal of this work is to explore the design of mechanisms for wide classes of succinctly representable mechanisms that would be robust against manipulations of the agents with bounded rationality and at the same time would return coalition structures of good quality in terms of the utilitarian social welfare. Following recent literature, we considered RAT and NOM mechanisms and proved that, while the former may fail to improve the inapproximability results for SP mechanisms, NOM mechanisms can guarantee both non-manipulability (under the assumption that agents lack contingent reasoning) and optimality (i.e., they return the best coalition structure) at least in the continuous case. Additionally, we proved that, if one requires computational efficiency as well, 
then we can compute non-obviously manipulable outcomes that approach the best possible approximation of the social welfare achievable in polynomial time. Moreover, we consider the~discrete variant and show that different transition phases occur with respect to the optimality of NOM mechanisms.

From a technical viewpoint, it would be clearly interesting to address the questions left open in this work, namely, to evaluate the approximability of the social welfare by RAT mechanisms for non-negative weights, to close the gap between the best known approximation achievable by a NOM mechanism and best possible polynomial-time approximation in FHGs, and to extend the optimality characterization to other discrete variants of ASHGs and FHGs.

Another intriguing future challenge is understanding to what extent the
results and techniques defined herein can be extended beyond ASHGs and FHGs, either to other classes of hedonic games, such as B-games or W-games, or to generalizations of hedonic games, such as the~group activity selection problem.
Moreover, considering 
the design of mechanisms that are not manipulable by bounded rational agents beyond NOM
%.
represents an important future work.
Indeed, as discussed above, NOM considers a specific way to aggregate the possible realizations of the~game for an agent lacking contingent reasoning. Hence, it is natural to ask what may happen with other ``simple'' aggregation rules,
or when a subset of competitors is able to reason contingently (as done for OSP in \cite{ferraioli2022obvious,arribillaga2024partition} or for RAT in \cite{hartman2025s}).

%\section*{Acknowledgments}
%This work was supported by: the PNRR MIUR project FAIR -- Future AI Research (PE00000013), Spoke 9 -- Green-aware AI; INdAM -- GNCS Project, code CUP\_E53C24001950001.

\bibliographystyle{plainnat}
\bibliography{references}

\newpage
\appendix
\section{Missing proofs}\label{sec:appendix}

\begin{proof}[Proof of \Cref{lemma:AlmostIsolated}]
    Let us consider an optimal partition $\pi^*$ and, for the sake of contradiction, let us assume that there exists a partition of $C$ into $k \geq 2$ disjoint subsets $C_1, \dots, C_k$ such that $C_h$ and $C_l$ are not in the same coalition in $\pi^*$ for every $h \neq l$. 
    
    Firstly, let us note that, if $i$ is not in the same coalition as $C_h$ for some $h \in [k]$, then $C_h \in \pi^*$, because $C$ is an almost isolated clique and therefore there are only negative connections between $C_h$ and any $C'' \subseteq \agents \setminus \left(C \cup \set{i}\right)$. Let us assume there exists a pair $h \neq l$ such that $C_h, C_l \in \pi^*$. Then, since there are only positive connections between $C_l$ and $C_h$, merging $C_h$ and $C_l$ into a single coalition strictly increases social welfare. Thus, there exists at most one $h \in [k]$ such that $C_h \in \pi^*$. Therefore, $k \leq 2$.

    Now, let us assume that $k = 2$. In this case, $i$ has to belong to the same coalition as either $C_1$ or $C_2$. Without loss of generality, let $C_1$ and $C_2 \cup \set{i} \cup C''$ belong to the optimal partition $\pi^*$ for some (possibly empty) $C'' \subseteq \agents\setminus\left(\set{i} \cup C\right)$.

    If $C'' \neq \emptyset$, replacing these coalitions with $C =  C_1 \cup C_2$ and $C'' \cup \set{i}$ will result into the change in the social welfare being at least equal to
\begin{align*}
    2|C_1||C_2| -2|C_2| + 2x|C_2||C''| = 2|C_2|(|C_1| - 1 + x|C''|) > 0,
\end{align*}
which contradicts the optimality of $\pi^*$.

If, instead, $C'' = \emptyset$, merging $C_1$ and $C_2 \cup \set{i}$ into a single coalition and thus deriving a new partition $\pi'$ will result in a change in the social welfare at least equal to
\begin{align*}
     2|C_1||C_2| - 2x|C_1| = 2|C_1|(|C_2| - x)
\end{align*}

If $x < 1$, then $2|C_1|(|C_2| - x) > 0$, which again contradicts the optimality of $\pi^*$. If, instead, $x = 1$, then $2|C_1|(|C_2| - x) \geq 0$, which means that, if $\pi^*$ is optimal, $\pi'$ has to be optimal as well.

Thus, in case $x < 1$, $k = 1$, while in case $x = 1$, there exists an optimum with $k = 1$, and the thesis follows.
\end{proof}

\begin{proof}[Proof of \Cref{lemma:OctopusOPT}]
    Let us note that $H$ is an almost isolated clique with the hinge node $i$. Then, from \Cref{lemma:AlmostIsolated}, unless $x = 1$, the possible optimum partitions are the following:
    \begin{enumerate}
        \item $\pi_1$ consisting of $H \cup \set{i}$ and other agents placed into singletons. Here, the social welfare is
        \begin{align*}
            |H|(|H| - 1) + |H| + f_i(H) - xe_i(H) = 
            |H|^2 + f_i(H) - xe_i(H),
        \end{align*}
        where $f_i(H)$ is the size of the set $H \cap \delta^+(i)$ and $e_i(H)$ is the size of the set $H \cap \delta^-(i)$.
        \item $\pi^{C}_2$ consisting of $H \cup \set{i} \cup C$ and other agents put into singletons for some $C \subseteq \agents\setminus\set{H \cup \set{i}}$ of size $c$. Here, the social welfare is at most
        \begin{align*}
            |H|(|H| - 1) + |H| + f_i(H) - xe_i(H) - 2x|H|c\ +
            c -cx -xc(c-1) \\= |H|^2 + f_i(H) - xe_i(H)+  c\left(1-2x|H|-xc\right).
        \end{align*}
        \item $\pi^{C'}_3$ consisting of $H$, $C' \cup \set{i}$ and other agents put into singletons for some $C' \subseteq \agents\setminus\set{H \cup \set{i}}$ of size $c'$. Here, the social welfare is at most
        \begin{align*}
            |H|(|H| - 1) + c' - c'x - xc'(c'-1) = |H|(|H| - 1) + c'(1-xc')
        \end{align*}
    \end{enumerate}
    If, instead, $x = 1$, \Cref{lemma:AlmostIsolated} does not guarantee that this is an exhaustive list of possible optimal partitions; however, at least one of the listed partitions still has to be optimal.

    We will next prove that $\pi_1$ is optimal. To this aim, it suffices to prove that the social welfare of $\pi_1$ is at least as large as the social welfare of $\pi_2^C$ and $\pi_3^{C'}$ for every $C$ and $C'$.
    \begin{align*}
        \SW(\pi_1) - \SW\left(\pi_2^{C}\right) = -c\left(1-2x|H|-xc\right) \geq 0 \Leftrightarrow |H| \geq \frac{1-xc}{2x}.
    \end{align*}
    Since $\frac{1-x}{2x} \geq \frac{1-xc}{2x}$ for every $c \geq 1$, this inequality holds whenever $|H| \geq \frac{1-x}{2x}$.
    \begin{align*}
        \SW(\pi_1) - \SW(\pi^{C'}_3) =  |H| + f_i(H) - xe_i(H) - c'(1-xc') \geq 0 \Leftrightarrow
        \\ |H| + f_i(H) - xe_i(H) \geq c'(1-xc')
    \end{align*}
    This inequality should hold for any $f_i(H)$, $e_i(H)$, and $1 \leq c' \leq n - |H| - 1$. So, since $f_i(H) \geq 0$ and $e_i(H) \leq |H|$, $\SW(\pi_1) \geq \SW(\pi^{C'}_3)$ for $x\neq 1$ as long as
    \begin{equation}
    \label{eq:oct_const}
        |H| - x|H| \geq c'(1-xc') \Leftrightarrow |H| \geq \frac{c'(1-xc')}{1-x}
    \end{equation}
    If $x \leq \frac{1}{2}$, then the right side of \eqref{eq:oct_const} is maximized when $c' = \frac{1}{2x} \geq 1$. Therefore, in this case, $\SW(\pi_1) \geq \SW(\pi^{C'}_3)$ as long as 
    \begin{align*}
        |H| \geq \frac{\frac{1}{2x}\left(1-\frac{x}{2x}\right)}{1-x} = \frac{1}{4x(1-x)}.
    \end{align*}
    However, since $|H| \geq \max\left\{\frac{1-x}{2x}; 1\right\} = \frac{1-x}{2x} > \frac{1}{4x(1-x)}$, this inequality always holds.
    
    If, instead, $x > \frac{1}{2}$, then the right side of \eqref{eq:oct_const} is maximized when $c' = 1$. So, in this case $\SW(\pi_1) \geq \SW(\pi^{C'}_3)$ as long as
    \begin{align*}
        |H| \geq \frac{1-x}{1-x} = 1,
    \end{align*}
    which holds, because $|H| \geq \max\left\{\frac{1-x}{2x}; 1\right\} = 1$.

    In the case $x = 1$, $|H| - x|H| \geq c'(1-xc')$ is equivalent to $0 \geq c'(1-c')$, which holds for any $c' \geq 1$.
\end{proof}

\begin{proof}[Proof of \Cref{lemma:duplexNoNegativeRelation}]
For the sake of a contradiction, let $C\in \pi$ be a coalition of size $c>1$ and let $i \in C$ be an agent such that $w_{ij} < 0$ for some $j \in C$. Then, $w(\cut(C\setminus \set{i}, \set{i})) <0$. In fact, we know that in the cut there are $2(c-1)$ edges. Among them, there is at least one negative edge $(i, j)$ whose value is at most $-2n +3$, and at most $2c -3$ positive edges, which have value exactly $1$. Hence, if $x > 2n-3$ or $c < n$, we have that  
\begin{align}
    w(\cut(C\setminus \set{i}, \set{i})) \leq 2c -3 -x < 0. \label{eq:negativeCutForDuplex}
\end{align}

Consider then the partition $\pi'$, which is exactly as $\pi$ except that $C$ is split into the coalitions $C\setminus\set{i}$ and $\set{i}$. Observe that $\SW(\pi) - \SW(\pi') = w(\cut(C\setminus \set{i}, \set{i}))$, and thus $\pi$ is strictly worse than moving $i$ to a singleton coalition, contradicting the optimality of $\pi$.

In the case that $x = 2n-3$ and $c = n$, we would have that $w(\cut(C\setminus \set{i}, \set{i})) = 0$, from which it follows that $\SW(\pi) - \SW(\pi') = 0$, i.e., both the grand coalition and the partition that moves a~single agent $i$ to another coalition are optimal. 
\end{proof}

\begin{proof}[Proof of \Cref{lemma:grandoralmost}]
Suppose instead that there exists an optimal outcome $\pi''\neq \pi, \pi'$. Let $C= \pi''(i)$. Clearly, $\pi''$ would not be optimal if it splits $\agents \setminus C$ in multiple coalitions $C_1, \ldots, C_\ell$, with $\ell \geq 2$, because all weights among agents in different coalitions are positive, and hence putting them in a single coalition will increase the social welfare by $\sum_{h \neq h'}2|C_h||C_{h'}| > 0$. Thus, it must be that $\pi''=\set{C, \agents\setminus C}$ and $C \neq \set{i}$. Let $c =\modulus{C} \geq 2$. Observe that $c < n-1$, otherwise $\SW(\pi) - \SW(\pi'') \geq 2(n-2) + 1 -x = 2n-3-x > 0$, which contradicts the optimality of $\pi''$. This implies that $n > 3$, otherwise the interval $[2, n-1)$ is empty. Then, we have
\begin{align*}
     \SW(\pi'') &\leq \SW(C) + \SW(\agents\setminus C) \\ & =  c(c-1) + (n-c)(n-c-1) = c^2+ (n-c)^2 - n\\
     &\leq 4 + (n-2)^2 - n = n^2-5n +8  \\ &< n^2 -3n +2= (n-1)(n-2)= \SW(\pi') \ . \qedhere
 \end{align*}
\end{proof}

\begin{proof}[Proof of \Cref{lem:kforxneqn-2}]
For a fixed an integer $k$, we derive $x< 2\frac{n-1}{k}-1$ from \Cref{eq:conditionPoitiveCut} and $x> \frac{n-1}{k}-1$ from \Cref{eq:conditionNegativeUtility}.
Therefore, they simultaneously hold true for $x \in (\frac{n-1}{k}-1, 2\frac{n-1}{k}-1)$; we denote such interval by $I_k$ as it depends on the value of $k$. Note that $I_1=(n-2, 2n-3)$ and $I_2=(\frac{n-1}{2}-1, n-2)$. The extreme points of the intervals $I_k$ are decreasing for increasing $k$. Furthermore, $2\frac{n-1}{k+1}-1 > \frac{n-1}{k}-1$ if and only if $k>1$, showing that the interval $I_{k+1}$, for $k\geq2$, overlaps $I_k$. Hence, for every $x \in (1, n-2) \cup (n-2, 2n-3)$, there is an interval $I_k$ such that $x \in I_k$, and thus 
\Cref{eq:conditionPoitiveCut,eq:conditionNegativeUtility}
can be simultaneously satisfied.
\end{proof}

\begin{proof}[Proof of \Cref{lem:final1n13}]
For the sake of contradiction, let us assume the opposite, i.e., that there is an instance for which \ref{mech:duplex} outputs an optimal partition $\hat{\pi}^*$ such that $i \in \hat{C}$, where $\hat{C} = E''_i \cup N''_i \cup F''_i$ for some $E''_i \subseteq E_i$, $N''_i \subseteq N_i$, $F''_i \subseteq F_i$, and $u_i(\hat{C}) < u_i(H)$, i.e. $|F''_i| - x|E''_i| < |F'_i| -x|E_i|$.
Note that this implies $|F''_i| < |F'_i|$, which is equivalent to $|F''_i| \leq |F'_i| - 1$ since the sizes of the sets can be only integer.

Let us assume that $\hat{\pi}^* = \set{\hat{C}, \hat{C}_1, \hat{C}_2, \dots, \hat{C}_k}$. Let $\hat{\pi}^*_l$ be a partition derived from $\hat{\pi}^*$ by merging $\hat{C}$ and $\hat{C}_l$ into a single coalition. Then, for $\hat{\pi}^*$ to be an optimum, $\SW(\hat{\pi}^*_l) - \SW(\hat{\pi}^*)$ should be non-positive for every $l \in [k]$, and thus
\[
    0 \geq \SW(\hat{\pi}^*_l) - \SW(\hat{\pi}^*) \geq -2x|\hat{C}|\cdot |\hat{C}_l| + |\hat{C}_l \cap F_i| - x|\hat{C}_l\cap E_i| - x|\hat{C}_l|.
\]
Hence,
\begin{equation*}
    |\hat{C}| \geq \frac{|\hat{C}_l \cap F_i| - x|\hat{C}_l\cap E_i|}{2x|\hat{C}_l|} - \frac{1}{2}.
\end{equation*}
Since this must hold for every $l \in [k]$, then in particular we have that
\begin{align*}
    |\hat{C}| & \geq \max_{l \in [k]}\left\{\frac{|\hat{C}_l \cap F_i| - x|\hat{C}_l\cap E_i|}{2x|\hat{C}_l|}\right\} - \frac{1}{2} \\
    & \geq \frac{\sum_{l \in [k]}\left(|\hat{C}_l \cap F_i| - x|\hat{C}_l\cap E_i|\right)}{2x\sum_{l \in[k]}|\hat{C}_l|} - \frac{1}{2} = \frac{|F_i| - |F''_i| - x(|E_i| - |E''_i|)}{2x(n - |\hat{C}| - 1)} - \frac{1}{2},
\end{align*}
which is, in turn, equivalent to:
\begin{equation*}
    x(n - |\hat{C}| - 1)(2|\hat{C}| + 1) \geq |F_i| - |F''_i| - x(|E_i| - |E''_i|).
\end{equation*}
Let $\Delta u = |F'_i| - |F''_i| - x(|E_i| - |E''_i|)$. Then, above inequality is equivalent to 
\[
    x(n - |\hat{C}| - 1)(2|\hat{C}| + 1) \geq |F_i| - x|E_i| - |F'_i| + x|E_i|  + \Delta u = |F_i| - |F'_i|  + \Delta u
\]
By the definition of $F'_i$, $|F_i| - |F'_i| = n - 1 - \left\lceil\frac{1-x}{2x}\right\rceil$. Thus:
\[
    x(n - |\hat{C}| - 1)(2|\hat{C}| + 1) \geq  n - 1 - \left\lceil\frac{1-x}{2x}\right\rceil + \Delta u
\]
or, equivalently
\begin{equation}
\label{eq:main_inequality}
    -2x|\hat{C}|^2 + x(2n -3)|\hat{C}| + x(n - 1) - (n - 1) + \left\lceil\frac{1-x}{2x}\right\rceil
    - \Delta u \geq 0
\end{equation}
We next prove that \eqref{eq:main_inequality} can never be satisfied, contradicting that $\SW(\hat{\pi}^*_l) - \SW(\hat{\pi}^*) \leq 0$, and thus completing the proof.

To this aim, let $D$ be the discriminant of the l.h.s. of \eqref{eq:main_inequality} with respect to $|\hat{C}|$. Note that
\begin{equation}
\label{eq:discriminant}
    D = x^2(2n - 3)^2 + 8x\left(x(n - 1) - (n - 1) + \left\lceil\frac{1-x}{2x}\right\rceil - \Delta u\right).
\end{equation}
This number is always positive for positive $n$ and $x$. Then, \eqref{eq:main_inequality} is satisfied as long as
\begin{equation}
\label{eq:main_solution}
    \frac{2n - 3}{4} - \frac{\sqrt{D}}{4x} \leq |\hat{C}| \leq \frac{2n - 3}{4} + \frac{\sqrt{D}}{4x}
\end{equation}
Let $\Delta c = \left\lceil\frac{1-x}{2x}\right\rceil - |\hat{C}|$. Then \eqref{eq:main_solution} implies that $\frac{2n - 3}{4} - \frac{\sqrt{D}}{4x} \leq \left\lceil\frac{1-x}{2x}\right\rceil - \Delta c$, or, equivalently,
\begin{equation}
    \label{eq:Dcondition}
    \sqrt{D} \geq x(2n - 3) - 4x\left\lceil\frac{1-x}{2x}\right\rceil + 4x\Delta c.
\end{equation}
Since $x > \frac{1}{n}$ for odd $n$ and $x > \frac{1}{n + 1}$ for even $n$, $\left\lceil\frac{1 - x}{2x}\right\rceil \leq \frac{n}{2}$. Therefore,
\[
    x(2n - 3) - 4x\left\lceil\frac{1-x}{2x}\right\rceil + 4x\Delta c \geq 2nx - 3x + 2xn + 4x = x > 0.
\]
Hence, we can square both sides in \eqref{eq:Dcondition}, and, by using \eqref{eq:discriminant}, we achieve that \eqref{eq:main_inequality} can only be satisfied if
\begin{align*}
 & x^2(2n - 3)^2 + 8x\left(x(n - 1) - (n - 1) + \left\lceil\frac{1-x}{2x}\right\rceil - \Delta u\right) \\
 & \geq \left(x(2n - 3) - 4x\left\lceil\frac{1-x}{2x}\right\rceil + 4x \Delta c\right)^2,
\end{align*}
or, equivalently,
\begin{equation}
\label{eq:final_cond}
    \begin{aligned}
     & xn - x - n + 1 + \left\lceil\frac{1-x}{2x}\right\rceil - \Delta u
    \geq 2x \left(\Delta c\right)^2 + 2x \left(\left\lceil\frac{1-x}{2x}\right\rceil\right)^2 \\
    & - 4x\left\lceil\frac{1-x}{2x}\right\rceil\Delta c + 2xn\Delta c - 3x\Delta c - 2x n \left\lceil\frac{1-x}{2x}\right\rceil + 3x\left\lceil\frac{1-x}{2x}\right\rceil
    \end{aligned}
\end{equation}

We next prove that \eqref{eq:final_cond} can never be satisfied and thus complete the proof. We distinguish between two cases.

First, let us consider the case when $\Delta c \leq \frac{1 + x}{x}$. Observe that, from
\begin{itemize}
    \item $\left\lceil\frac{1-x}{2x}\right\rceil = |H| = |E_i| + |N_i| + |F'_i|$,
    \item $|\hat{C}| = |E''_i| + |N''_i| + |F''_i|$,
    \item $|F''_i| \leq |F'_i| - 1$,
    \item and $|N_i| \geq |N_i''|$,
\end{itemize} it follows that 
\begin{equation}
\label{eq:Deltac}
    \Delta c \geq |E_i| - |E''_i| + 1,
\end{equation}
and then $\Delta u \geq 1 + x - x\Delta c$. Hence, \eqref{eq:final_cond} can be satisfied only if
\begin{align*}
 & xn - x - n + 1 + \left\lceil\frac{1-x}{2x}\right\rceil - 1 - x + x\Delta c
    \geq 2x \left(\Delta c\right)^2 + 2x \left(\left\lceil\frac{1-x}{2x}\right\rceil\right)^2 \\
& - 4x\left\lceil\frac{1-x}{2x}\right\rceil\Delta c + 2xn\Delta c - 3x\Delta c - 2x n \left\lceil\frac{1-x}{2x}\right\rceil + 3x\left\lceil\frac{1-x}{2x}\right\rceil,   
\end{align*}
or equivalently
\begin{equation}
\label{eq:final_simplified_case1}
\begin{aligned}
& - 2x \left(\Delta c\right)^2 + \left(x + 4x\left\lceil\frac{1-x}{2x}\right\rceil - 2xn + 3x\right)\Delta c \\
& \geq
    - xn + 2x + n - \left\lceil\frac{1-x} {2x}\right\rceil + 2x \left(\left\lceil\frac{1-x}{2x}\right\rceil\right)^2 - 2x n \left\lceil\frac{1-x}{2x}\right\rceil + 3x\left\lceil\frac{1-x}{2x}\right\rceil
\end{aligned}
\end{equation}
Note that, from \eqref{eq:Deltac} and $|E_i| \geq |E''_i|$, it follows that $\Delta_c \geq 1$, and therefore the left side of \eqref{eq:final_simplified_case1} is maximized at $\Delta c = \max\{1, \frac{4x + 4x\left\lceil\frac{1-x}{2x}\right\rceil - 2xn}{4x} = 1 + \left\lceil\frac{1-x}{2x}\right\rceil - \frac{n}{2}\} = 1$, where we used that $\left\lceil\frac{1-x}{2x}\right\rceil \leq \frac{n}{2}$. Hence, \eqref{eq:final_cond} can be satisfied only if \eqref{eq:final_simplified_case1} holds even we replace $\Delta c$ with $1$, i.e.,
\[
    x\left\lceil\frac{1-x}{2x}\right\rceil - xn \geq
     n - \left\lceil\frac{1-x} {2x}\right\rceil + 2x \left(\left\lceil\frac{1-x}{2x}\right\rceil\right)^2 - 2x n \left\lceil\frac{1-x}{2x}\right\rceil
\]
or, equivalently,
\[
     \left(n - \left\lceil\frac{1-x} {2x}\right\rceil\right)\left(1 + x - 2x\left\lceil\frac{1-x} {2x}\right\rceil\right) \leq 0,
\]
that is impossible, since $\left\lceil\frac{1-x} {2x}\right\rceil \leq \frac{n}{2}$ (and thus $n - \left\lceil\frac{1-x} {2x}\right\rceil > 0$) and $\left\lceil\frac{1-x} {2x}\right\rceil < \frac{1-x}{2x} + 1 = \frac{1+x}{2x}$ (and thus $1 + x - 2x\left\lceil\frac{1-x} {2x}\right\rceil > 0$).

Consider now the case that $\Delta c > \frac{1 + x}{x}$.
Note that $\Delta u = (|F'_i| - x|E_i|) - (|F''_i| - x|E''_i|) = u_i(H) - u_i(\hat{C})$, that is, by construction, strictly positive. Hence, \eqref{eq:final_cond} can be satisfied only if
\begin{align*}
 & xn - x - n + 1 + \left\lceil\frac{1-x}{2x}\right\rceil
    > \frac{2(1+2x+x^2)}{x} + 2x \left(\left\lceil\frac{1-x}{2x}\right\rceil\right)^2 \\
    & - 4(1+x)\left\lceil\frac{1-x}{2x}\right\rceil + 2(1+x)n - 3(1+x) - 2x n \left\lceil\frac{1-x}{2x}\right\rceil + 3x\left\lceil\frac{1-x}{2x}\right\rceil,
\end{align*}
or, equivalently,
\[
 3n + xn + \frac{2}{x} + 2x \left(\left\lceil\frac{1-x}{2x}\right\rceil\right)^2 - \left(5 + x + 2xn\right)\left\lceil\frac{1-x}{2x}\right\rceil < 0.
\]
However, this is impossible. Indeed, by using that $\frac{1 - x}{2x} \leq \left\lceil\frac{1-x}{2x}\right\rceil < \frac{1 + x}{2x}$, we have that
\begin{align*}
    & 3n + xn + \frac{2}{x} + 2x \left(\left\lceil\frac{1-x}{2x}\right\rceil\right)^2 - \left(5 + x + 2xn\right)\left\lceil\frac{1-x}{2x}\right\rceil >\\
    & 3n + xn + 2x + \frac{1 - 2x + x^2}{2x} - (5 + x + 2xn) \cdot \frac{1+x}{2x}
    = 2n - 4 \geq 0. \qedhere
\end{align*}
\end{proof}

\begin{lemma}
\label{lemma:all_ones_manipulation}
Consider an instance in which there is an agent $i \in \agents$ such that $w_{ij} = 1$ for every $j \in \agents\setminus\set{i}$, and let $\pi$ be the optimal partition for this instance.
Then either $\pi = \set{\agents}$ or $|\pi(i)| \geq \frac{1+x}{2x}$.
\end{lemma}
\begin{proof}
    Suppose this is not the case, i.e., $\pi \neq \set{\agents}$ and $|\pi(i)| < \frac{1+x}{2x}$.
    Let us pick an arbitrary coalition $C' \in \pi\setminus\pi(i)$ (since $\pi \neq \set{\agents}$, this is always possible). For every $j \in C'$, $w_{ij} = 1$;
    therefore, there are at least $|C'|$ positive relations between $\pi(i)$ and $C'$.
    Then, for the partition $\pi'$ derived from $\pi$ by merging $\pi(i)$ and $C'$ into a single coalition:
    \begin{align*}
        \SW(\pi') - \SW(\pi) & \geq (1-x)|C'| - 2x|C'|\left(|\pi(i)| - 1\right) \\
        & = |C'|\cdot \left(1 - x - 2x|\pi(i)| + 2x\right) = 2x|C'|\cdot\left(\frac{1+x}{2x} - |\pi(i)|\right) > 0.
    \end{align*}
    and thus $\SW(\pi') > \SW(\pi)$, which means that $\pi$ cannot be optimal.
\end{proof}

\begin{proof}[Proof of Lemma~\ref{lemma:WeaklyConnectingAgent}]
Consider an instance in which there is an agent $i \in \agents$ such that $w_{ij} = 1$ for every $j \in \agents\setminus\set{i}$, and let $\pi$ be the optimal partition for this instance.
Suppose that $\pi = \set{C_1, \dots, C_k}$. If $\pi$ is not the grand coalition, then $|\pi(i)| \leq n - 1$. However, by \Cref{lemma:all_ones_manipulation}, $|\pi(i)| \geq \frac{1+x}{2x}$, and thus
    \begin{align*}
        |\pi(i)| \geq \frac{1+x}{2x} = \frac{1}{2x} + \frac{1}{2} > \frac{1}{2\cdot\frac{1}{2n - 3}} + \frac{1}{2} = \frac{2n - 3 + 1}{2} = n - 1,
    \end{align*}
    which contradicts $|\pi(i)| \leq n - 1$. Therefore, $\pi(i) = \agents$.
\end{proof}

\begin{proof}[Proof of \Cref{lem:weakly_connected_sp}]
Let $G^+$ denote the directed graph induced by the positively weighted edges between the agents, and let $\pi$ be the partition of $G^{+}$ into weakly connected components. Firstly, let us prove that $\pi$ is optimal. Note that any other outcome can be derived from $\pi$ by either splitting a weakly connected component of $G^+$ into multiple coalitions, or merging agents belonging to different weakly connected components in the same coalition. We next show that both these operations cannot increase the social welfare.

Let us consider first the merging operation. Since, by the definition of $G^+$, all preferences between agents belonging to two different weakly connected components of $G^+$ are non-positive, it is not possible to increase the social welfare of a partition by joining the agents from different weakly connected components of $G^+$ into a new coalition.

We next prove that it is not possible to improve the social welfare of a partition by splitting a weakly connected component of $G^+$ into multiple coalitions.
Specifically, let $C \subseteq \agents$ be a weakly connected component of $G^+$ and let $\set{C_1, \ldots, C_m}$, for $m \geq 2$, be an arbitrary partition of $C$. Let $\pi$ and $\pi'$ be two partitions such that only differ for agents in $C$, for which $\pi$ returns the coalition $C$, and $\pi'$ returns the coalitions $C_1, \ldots, C_m$. We show that, if $x \in \left(0; \frac{2}{n^2 - 3 + \parity}\right]$, then $\SW(\pi) \geq \SW(\pi')$.
    Note indeed that $\SW(\pi') - \SW(\pi) = -(x+1)(m-1)+x\sum_{h \neq \ell} |C_\ell||C_h|$.
    Since $|C_1|, \dots, |C_m| \in \mathbb{N}$, then
    \begin{align*}
        \sum_{h \neq l}|C_l||C_h| &\leq \left(|C| \mod m\right)\left(\left(|C| \mod m\right) - 1\right)\left\lceil\frac{|C|}{m}\right\rceil^2\\
        & \qquad + 2\left(|C| \mod m\right)\left(m - \left(|C| \mod m\right)\right)\left\lceil\frac{|C|}{m}\right\rceil\left\lfloor\frac{|C|}{m}\right\rfloor\\
        & \qquad + \left(m - \left(|C| \mod m\right)\right)\left(m - \left(|C| \mod m\right) - 1\right)\left\lfloor\frac{|C|}{m}\right\rfloor^2
        \\ & \leq (m - 1)(m - 2)\left\lceil\frac{|C|}{m}\right\rceil^2 + 2(m-1)\left\lceil\frac{|C|}{m}\right\rceil\left\lfloor\frac{|C|}{m}\right\rfloor,
    \end{align*}
    where the last inequality is tight if $|C| \mod m = 0$, and it holds since $\left\lfloor\frac{|C|}{m}\right\rfloor < \left\lceil\frac{|C|}{m}\right\rceil$ otherwise.
    Therefore:  
    \begin{align*}
        \SW(\pi') - \SW(\pi) & \leq -(x+1)(m - 1) + x\left((m-1)(m-2)\left\lceil\frac{|C|}{m}\right\rceil^2 + 2(m - 1)\left\lceil\frac{|C|}{m}\right\rceil\left\lfloor\frac{|C|}{m}\right\rfloor\right) \\
        & \leq -x-1 + 2x\left\lceil\frac{|C|}{2}\right\rceil\cdot\left\lfloor\frac{|C|}{2}\right\rfloor \leq - x - 1 + 2x\left\lceil\frac{n}{2}\right\rceil\cdot\left\lfloor\frac{n}{2}\right\rfloor,
    \end{align*}
    where the second inequality holds since $\left\lceil\frac{|C|}{m}\right\rceil \leq 2 \left\lfloor\frac{|C|}{m}\right\rfloor$, $\left\lfloor\frac{|C|}{m}\right\rfloor \leq \frac{2}{m} \left\lfloor\frac{|C|}{2}\right\rfloor$, $\left\lceil\frac{|C|}{m}\right\rceil \leq \frac{2}{m} \left\lceil\frac{|C|}{2}\right\rceil$, and, for $m \geq 2$, $\frac{8(m-1)^2}{m^2} \leq 2$.
    If $n$ is even, then we have
     \begin{align*}
        \SW(\pi') - \SW(\pi) & \leq x\left(\frac{n^2}{2} - 1\right) - 1 \leq \frac{2}{n^2 - 2}\cdot\frac{n^2 - 2}{2} - 1 = 0.
    \end{align*}
    Similarly, if $n$ is odd,
    \begin{align*}
         \SW(\pi') - \SW(\pi) & \leq x\left(\frac{n^2 - 1}{2} - 1\right) - 1 \leq \frac{2}{n^2 - 3}\cdot\frac{n^2 - 3}{2} - 1 = 0. \qedhere
    \end{align*}

Therefore, $\pi$ is optimal, and it remains to show that it is the same outcome that \ref{mech:duplex} chooses according to $\textsc{TieBreaking}(x)$. Clearly, the partition of the agents into weakly connected components of $G^+$ is by definition a partition without positively weighted edges outside the coalitions (every positive edge is within some weakly connected component). Moreover, since splitting a weakly connected component into multiple coalitions will create at least one positively weighted edge outside the coalitions, there cannot be another optimal partition with a larger number of coalitions maintaining the property of having all positive edges within the coalitions.
\end{proof}

\subsection{Missing part of Case~\ref{subsec:case_splitted}}
\begin{lemma}\label{corollary:manipulation_all_ones}
    Let $x < \frac{1}{n-2}$ and $C^{\leq 0}_i = \set{j \in \agents\setminus\set{i}\ |\ w_{ij} \leq 0}$. If $|C_i^{\leq 0}| < \frac{1-x}{2x}$, then, by declaring $d_{ij} = 1$ for every $j \in \agents\setminus\set{i}$, the agent $i$ can guarantee that $u_i(\pi^*) > 0$ for any optimal partition $\pi^*$.
\end{lemma}
\begin{proof}
    Since $d_{ij} = 1$ for every $j \in \agents\setminus\set{i}$, from \Cref{lemma:all_ones_manipulation}, in an optimal partition $\pi^*$, $i$ must be in a coalition of size less than $\frac{1+x}{2x}$. Therefore, since $\left|C^{\leq 0}_i\right| < \frac{1-x}{2x} = \frac{1+x}{2x} - 1$, the coalition $\pi^*(i)$ has to contain members other than $i$ and the agents from $C^{\leq 0}_i$. Therefore, there exists $j \in \pi^*(i)$ such that $w_{ij} = 1$, and $u_i(\pi^*) \geq 1 - x(n-2) > 1 - \frac{n - 2}{n - 2} = 0$.
\end{proof}

    Let us now consider the following (true) weights of the agent $1$:
    \begin{itemize}
        \item $w_{1j} = -x$ for $2 \leq j \leq \left\lceil\frac{1-x}{2x}\right\rceil$;
        \item $w_{1j} = 1$ for $\left\lceil\frac{1-x}{2x}\right\rceil + 1 \leq j \leq n$.
    \end{itemize}
    Then, let us consider the instance where for $C_1 = \set{2, \dots, \left\lfloor\frac{n}{2}\right\rfloor}$ and $C_2 = \agents\setminus\left(C_1 \cup \set{1}\right)$:
    \begin{itemize}
        \item for every $i, j \in C_1$, $w_{ij} = w_{ji} = w_{i1} = 1$;
        \item for every $i, j \in C_2$, $w_{ij} = w_{ji} = 1$;
        \item for every $i \in C_1$ and $j \in C_2$, $w_{ij} = w_{ji} = w_{j1} = -x$.
    \end{itemize}
    We next prove that the unique optimal partition for this instance is $\pi^* = \set{C_1 \cup \set{1}, C_2}$.
    Indeed, in this instance, both $C_1$ and $C_2$ are almost isolated cliques with the hinge node $1$. Therefore, by Lemma~\ref{lemma:AlmostIsolated}, apart from $\pi^*$ there are only three other options for the optimal partition, which we next show are worse than $\pi^*$.
    \begin{description}
        \item[$\pi_1 = \set{\agents}$:]
    \begin{align*}
        \SW(\pi_1) - \SW(\pi^*) &= (1-x)\left(n - \left\lceil\frac{1-x}{2x}\right\rceil\right) - 2x|C_1||C_2| - 2x\left(|C_2| - n + \left\lceil\frac{1-x}{2x}\right\rceil\right) \\
        &= -(1 + x)\left\lceil\frac{1-x}{2x}\right\rceil + (1 + x)n - 2x\left\lfloor\frac{n}{2}\right\rfloor\left\lceil\frac{n}{2}\right\rceil\\
        & \leq \frac{x^2-1}{2x} + (1+x)n - 2x\left\lfloor\frac{n}{2}\right\rfloor\left\lceil\frac{n}{2}\right\rceil
    \end{align*}
    For even $n$:
    \begin{align*}
        \SW(\pi_1) - \SW(\pi^*) &=  \frac{x^2 - 1 + 2nx + 2nx^2 - x^2n^2}{2x} = \frac{(1 + 2n - n^2)x^2 + 2nx - 1}{2x} < 0,
    \end{align*}
    where we used that $X = (1 + 2n - n^2)x^2 + (2n - 2)x - 1$ (strictly) increases for $x < \frac{2n}{2(n^2 - 2n - 1)}$, and, since $x < \frac{1}{n + \sqrt{2n + 1}} < \frac{2n}{2(n^2 - 2n - 1)}$ (where the latter holds since $n \geq 3$), then $X$ achieves its maximum at $x \rightarrow \frac{1}{n + \sqrt{2n + 1}} - 0$ from which we derive that
    \begin{align*}
        X & < \frac{1+2n - n^2}{(n + \sqrt{2n + 1})^2} + \frac{2n}{n+\sqrt{2n + 1}} - 1 \\
        & = \frac{1 + 2n - n^2 + 2n^2 + 2n\sqrt{2n + 1}- n^2 - 2n\sqrt{2n + 1} - 2n - 1}{(n + \sqrt{2n + 1})^2} = 0.
    \end{align*}
    For odd $n$:
    \begin{align*}
        \SW(\pi_1) - \SW(\pi^*) & = \frac{x^2 - 1 + 2nx + 2nx^2 - x^2n^2 + x^2}{2x} = \frac{(2 + 2n - n^2)x^2 + 2nx - 1}{2x} < 0,
    \end{align*}
    where, similarly to the previous case, we used that $X = (2 + 2n - n^2)x^2 + 2nx - 1$ is maximized at $x \rightarrow \frac{1}{n+\sqrt{2n + 2}} - 0$, and thus
    \begin{align*}
        X & < \frac{2+2n - n^2}{(n+\sqrt{2n + 2})^2} + \frac{2n}{n+\sqrt{2n + 2}} - 1\\
        & = \frac{2 + 2n - n^2 + 2n^2 + 2n\sqrt{2n + 2} - n^2 - 2n\sqrt{2n + 2} - 2n - 2}{(n+\sqrt{2n + 2})^2} = 0.
    \end{align*}
    
        \item[$\pi_2 = \set{C_1, C_2 \cup \set{1}}$:]
    \begin{align*}
        \SW(\pi_2) - \SW(\pi^*) & = (1-x)\left(n - \left\lceil\frac{1-x}{2x}\right\rceil\right) - 2x\left(|C_2| - n + \left\lceil\frac{1-x}{2x}\right\rceil\right) - (1-x)|C_1| \\
        & = (1+x)n - (1+x)\cdot \left\lceil\frac{1-x}{2x}\right\rceil - 2x(n - |C_1| - 1) - (1-x)|C_1| \\
        & = (1 - x) n - (1+x)\cdot \left\lceil\frac{1-x}{2x}\right\rceil + 2x - (1 - 3x)|C_1| \\
        & = (1 - x) n - (1+x)\cdot \left\lceil\frac{1-x}{2x}\right\rceil + 2x - (1 - 3x)\left(\left\lfloor\frac{n}{2}\right\rfloor - 1\right) \\
        & = (1 - x) n + (1 - x) - (1+x)\cdot \left\lceil\frac{1-x}{2x}\right\rceil  - (1 - 3x)\left\lfloor\frac{n}{2}\right\rfloor\\
        & \leq (1 - x) n + (1 - x) - (1+x)\frac{1-x}{2x}  - (1 - 3x)\left\lfloor\frac{n}{2}\right\rfloor
    \end{align*}
    If $n$ is even, then
    \begin{align*}
        \SW(\pi_2) - \SW(\pi^*) & \leq (1 - x) n + (1 - x) - (1+x)\cdot \frac{1-x}{2x} - (1 - 3x)\frac{n}{2} \\
        & = \frac{2xn - 2x^2n + 2x - 2x^2 + x^2 - 1 - x(1-3x)n}{2x} \\
        & = \frac{(-2n - 1)x^2 + (2 + 2n)x - 1 - xn + 3x^2n}{2x} \\
        & = \frac{(n - 1)x^2 + (n + 2)x - 1}{2x} < 0,
    \end{align*}
    where the last inequality follows since $X = (n - 1)x^2 + (n + 2)x - 1$ increases when $x > \frac{-n - 2}{2(n - 1)}$, which always holds since $n \geq 4$, and, therefore, $X$ is maximized at $x \rightarrow \frac{1}{n + \sqrt{2n + 1}} - 0$, from which
    \begin{align*}
        X &< \frac{n - 1}{\left(n + \sqrt{2n + 1}\right)^2} + \frac{n + 2}{n + \sqrt{2n + 1}} - 1 \\
        & = \frac{n - 1 + n^2 + 2n + n\sqrt{2n + 1} + 2\sqrt{2n + 1} - n^2 - 2n\sqrt{2n + 1} - 2n - 1}{\left(n + \sqrt{2n + 2}\right)^2} \\
        & = \frac{(2-n)\sqrt{2n + 2} + n}{\left(n + \sqrt{2n + 2}\right)^2} < 0,
    \end{align*}
    where the last inequality holds for any $n \geq 4$.
    
    Similarly, if $n$ is odd, then
    \begin{align*}
        \SW(\pi_2) - \SW(\pi^*) & \leq (1 - x) n + (1 - x) - (1+x)\cdot \frac{1-x}{2x} - (1 - 3x)\frac{n-1}{2} \\
        & = \frac{2xn - 2x^2n + 2x - 2x^2 + x^2 - 1 - x(1-3x)(n - 1)}{2x} \\
        & = \frac{(-2n - 1)x^2 + (2 + 2n)x - 1 - xn + x + 3x^2n - 3x^2}{2x} \\
        & = \frac{(n - 4)x^2 + (n + 3)x - 1}{2x} < 0,
    \end{align*}
    where the last inequality follows since $X = (n - 4)x^2 + (n + 3)x - 1$ increases when $x > \frac{-n - 3}{2(n - 4)}$, which always holds since $n \geq 5$, and, therefore, $X$ is maximized at $x \rightarrow \frac{1}{n + \sqrt{2n + 2}} - 0$, from which
    \begin{align*}
        X &< \frac{n - 4}{\left(n + \sqrt{2n + 2}\right)^2} + \frac{n + 3}{n + \sqrt{2n + 2}} - 1 \\
        & = \frac{n - 4 + n^2 + 3n + n\sqrt{2n + 2} + 3\sqrt{2n + 2} - n^2 - 2n\sqrt{2n + 2} - 2n - 2}{\left(n + \sqrt{2n + 2}\right)^2} \\
        & = \frac{(3-n)\sqrt{2n + 2} + 2n - 6}{\left(n + \sqrt{2n + 2}\right)^2} < 0,
    \end{align*}
    where the last inequality holds for any $n \geq 5$.
        \item[$\pi_3 = \set{C_1, C_2, \set{1}}$:]
    \begin{align*}
        \SW(\pi_3) - \SW(\pi_1) = -(1-x)|C_1| < 0
    \end{align*}
    \end{description}
    Therefore, $\pi^*$ is the unique optimum for this instance, and the agent $1$ gets utility $-x|C_1| < 0$. Hence, the worst-case utility of agent $1$ is negative.
    On the other hand, if $1$ declares $d_{1i} = 1$ for each $i \in \agents\setminus\set{1}$, then from \Cref{corollary:manipulation_all_ones} it follows that, in any optimal outcome, the utility of $1$ will be greater than zero since $C_1^{\leq 0} = \left\lceil\frac{1 - x}{2x}\right\rceil - 1 < \left\lceil\frac{1 - x}{2x}\right\rceil$. Thus, this declaration increases the worst-case utility of $1$, and hence any optimal mechanism is obviously manipulable.
\end{document}